\documentclass{article}
\usepackage{natbib}

\usepackage{algorithmic} 
\usepackage{algorithm} 
\usepackage{setspace} 
\usepackage{rotating} 
\usepackage[margin=1in]{geometry} 
\usepackage{amsthm} 
\usepackage{amsfonts, amssymb} 
\usepackage{epstopdf} 
\usepackage{hyperref} 
\usepackage{bm} 
\usepackage{authblk} 
\usepackage{multirow} 
\usepackage{latexsym} 
\usepackage{rotating} 
\usepackage{titlesec} 
\usepackage{caption} 
\usepackage{mathtools} 
\usepackage{color} 
\usepackage{subcaption}

\newtheorem{theorem}{Theorem}
\newtheorem{corollary}{Corollary}
\newtheorem{lemma}{Lemma}

\theoremstyle{definition}
\newtheorem{definition}{Definition}

\theoremstyle{remark}
\newtheorem{remark}{Remark}[section]

\newcommand{\iid}{\stackrel{\mathrm{i.i.d.}}{\sim}}

\newcommand{\amin}[1]{\underset{#1}{\operatorname{argmin~}}} 

\newcommand{\inP}{\stackrel{p}{\to}}

\newcommand{\reals}{\mathbb{R}}

\def\spacingset#1{\renewcommand{\baselinestretch}%
{#1}\small\normalsize} \spacingset{1}

\usepackage{authblk}

  \title{\bf Inference After Selecting Plausibly Valid Instruments with Application to Mendelian Randomization}
\author[1]{Nan Bi}
\author[2]{Hyunseung Kang} 
\author[1]{Jonathan Taylor}
\affil[1]{Department of Statistics, Stanford University}
\affil[2]{Department of Statistics, University of Wisconsin-Madison}
\date{}                     
	
\begin{document}
\maketitle

\bigskip
\begin{abstract}
Mendelian randomization (MR) is a popular method in genetic epidemiology to estimate the effect of an exposure on an outcome by using genetic instruments. These instruments are often selected from a combination of prior knowledge from genome wide association studies (GWAS) and data-driven instrument selection procedures or tests. Unfortunately, when testing for the exposure effect, the instrument selection process done a priori is not accounted for. This paper studies and highlights the bias resulting from not accounting for the instrument selection process by focusing on a recent data-driven instrument selection procedure, sisVIVE, as an example. We introduce a conditional inference approach that conditions on the instrument selection done a priori and leverage recent advances in selective inference to derive conditional null distributions of popular test statistics for the exposure effect in MR. The null distributions can be characterized with individual-level or summary-level data in MR. We show that our conditional confidence intervals derived from conditional null distributions attain the desired nominal level while typical confidence intervals computed in MR do not. We conclude by reanalyzing the effect of BMI on diastolic blood pressure using summary-level data from the UKBiobank that accounts for instrument selection.
\end{abstract}

\noindent%

\bigskip
\spacingset{1.5} 

\section{Introduction} \label{sec:intro}

\subsection{Motivation: Mendelian Randomization and Selecting Instruments} 

Mendelian randomization (MR) \citep{davey_smith_mendelian_2003, davey_smith_mendelian_2004,burgess_mendelian_2015} has become a popular tool to estimate the effect of an exposure or a treatment on an outcome. In a nutshell, MR utilizes instrumental variables (IV), a popular method in causal inference, epidemiology, and economics \citep{angrist_instrumental_2001,hernan_instruments_2006,baiocchi_instrumental_2014} and large genome wide association studies (GWAS) to find genetic variants, typically single nucleotide polymorphisms (SNPs) and called genetic instruments, that satisfy the three core assumptions: 
\begin{itemize}
\item[(A1)] the instruments are associated with the treatment
\item[(A2)] the instruments have no direct effect on the outcome conditional on the treatment value
\item[(A3)] the instruments are unconfounded;
\end{itemize}
See Section \ref{sec:def} for formal mathematical definitions in our model. If the magnitude of the association in (A1) is large, instruments are said to be strong and if the magnitude is small, instruments are said to be weak. If instruments satisfy (A2) and (A3), they are referred to as valid instruments \citep{murray_avoiding_2006}. 

A standard MR analysis starts by selecting SNPs/instruments that satisfy the three assumptions, especially the first two assumptions (A1) and (A2); (A3) is generally assumed to hold in MR studies \citep{burgess_mendelian_2015}. Selecting instruments that satisfy (A1) is typically based on prior GWAS that show strong links between the SNPs and the exposure of interest. Selecting instruments that satisfy (A2) requires SNPs to only affect the outcome via the exposure; SNPs must not have pleiotropic effects or affect other variables which affect the outcome \citep{solovieff_pleiotropy_2013}. Unfortunately, finding non-pleiotropic instruments is a challenge in MR, especially if the biomarker and/or the outcome are complex traits \citep{little_mendelian_2003, thomas_commentary_2004, brennan_commentary_2004, lawlor_mendelian_2008}, and many MR investigators use a combination of subject-matter expertise, prior GWAS, and ad-hoc association tests, to assess (A2). More recently, data-driven methods have been proposed to select instruments that plausibly satisfy (A2) under the assumption of a linear MR model \citep{andrews_consistent_1999, kang_instrumental_2016,windmeijer2016use,guo2016causalci}; see Section \ref{sec:def} for details. Once instruments are selected to plausibly satisfy (A1) to (A3), the effect of the exposure on the outcome is estimated and tested. The two most common inferential goals are  (i) testing the null hypothesis of no effect $H_0: \beta^* = 0$ where $\beta^*$ is the parameter for the effect of the exposure on the outcome and (ii) constructing a $95\%$ confidence interval of $\beta^*$.


As a concrete example of a standard MR analysis, \citet{vimaleswaran_causal_2013} studied the causal effect of obesity, as measured by the body mass index (BMI), on serum vitamin D levels. To satisfy (A1),  the authors selected 12 SNPs strongly associated to BMI based on prior GWAS by \citet{loos_common_2008}, \citet{thorleifsson_genome-wide_2008}, \citet{li_cumulative_2009}, and \citet{speliotes_association_2010}. They also tested the 12 SNPs for violations of (A2) by conducting an association test between each SNP and the outcome; see Table S6 of \citet{vimaleswaran_causal_2013}. 
After these checks, the 12 SNPs were then used to infer BMI's effect on vitamin D by using a variant of two-stage least squares (TSLS), a popular method in IV; see \citet{lawlor_mendelian_2008} and Section \ref{sec:point_est} for details on TSLS in MR settings. They reported a p-value of $0.005$ for the null hypothesis of no effect of BMI on serum vitamin D levels with a 95\% confidence interval of $(-0.71, -0.13)$ and concluded that BMI is negatively associated with serum vitamin D levels.

\subsection{Inferring Treatment Effects After Selection} 
\label{sec:intro_condinf}
Unfortunately, the standard approach of inferring treatment effects and deriving p-values for the null of no treatment effect neglect the instrument selection process done a priori. In MR, the SNPs/instruments were selected (or discarded) from a large pool of potential SNPs from GWAS whereby the ``best'' instruments were selected and the ``worst'' instruments were discarded based on a combination of subjective-matter expertise and quantitative analysis. Also, in many MR studies, including \citet{vimaleswaran_causal_2013}, the same data is used to select the instruments as well as to test the exposure effect. Yet, when testing this effect by computing the p-value or the confidence interval, the statistical analysis simply assumes that the selected instruments are given ``as is'' and the prior selection process is neglected. Formally speaking,  the p-value or the confidence interval for the exposure effect is not conditional on the event that the ``best'' instruments were selected into the analysis.

For example, going back to  \citet{vimaleswaran_causal_2013}, the final p-value $0.005$ of the exposure effect and the 95\% confidence interval were computed under the assumption that the 12 selected SNPs were drawn from a population of exactly 12 instruments. In reality, the 12 SNPs are the ``best'' SNPs from a larger pool of SNPs from GWAS and they  were carefully selected for IV strength and validity, all in order to obtain the ``best'' exposure effect. Consequently, when conditioning on the selection process, the authors'  significant p-value may over-represent the true statistical significance of the treatment effect when none actually exist and increase Type I error. Similarly, the confidence intervals may be optimistically too short.

\subsection{Prior Work} 
Informally, the problem above goes by many names such as the ``file drawer effect'' \citep{fithian_optimal_inference_2014,taylor_statistical_2015} or ``p-hacking'' \citep{simmons_false_2011,gelman_garden_2013}, where researchers only report the final statistical analysis and ignore prior analysis, such as variable selection, that led up to the final analysis. Some recent work have framed the problem as {\emph{selective inference}} and showed that ignoring variable selection when testing a hypothesis can lead to inflated Type I errors and biased confidence intervals \citep{fithian_optimal_inference_2014,lee_marginal_screening_2014,taylor_forward_stepwise_2014,tian_magic_2016,bi_inferactive_2017}. In Section \ref{sec:simulation}, we also show that this phenomena holds in MR, where the typical confidence interval from MR ignoring instrument selection can have coverage much lower than the nominal level in many cases. 

To the best of our knowledge, none of the work in MR considered testing the treatment effect that takes into account instrument selection. Much of the recent work in MR has focused on estimation when some instruments may violate the IV assumptions \citep{bowden_mendelian_2015,kang2015simple,bowden_consistent_2016, burgess_robust_2016,windmeijer2016use,kang_instrumental_2016,guo2016causalci,hartwig2017robust}. 
Works by \citet{small_sensitivity_2007} and \citet{conley_plausibly_2012} discuss sensitivity analysis as a way to conduct inference when the selected instruments are invalid. However, none of these work addressed the issue of calibrating inference of the exposure effect after instrument selection has taken place.

\subsection{Our Contribution}
In this paper, we propose a method to generate \emph{honest} p-values and confidence intervals for the treatment effect. Our method is honest in that it accounts for the instrument selection process when computing statistical significance of the exposure effect. We focus on an instrument selection procedure by \citet{kang_instrumental_2016} called sisVIVE, which operationalized the MR investigator's instrument selection process by choosing $s$ instruments from a pool of $L$ candidate instruments based on a procedure similar to the Lasso \citep{tibshirani_regression_1996}; sisVIVE was further analyzed by \citet{windmeijer2016use} for its selection properties. We propose a flexible and general sampling method to generate the conditional null distribution (i.e. conditional on selecting the instrument(s)) of pre-existing test statistics $T$ for the exposure effect $\beta^*$. We provide two examples of our sampling method by deriving the conditional distribution of two popular test statistics, the aforementioned TSLS, which is the most popular test statistic in MR, and the Anderson-Rubin test statistic \citep{anderson_estimation_1949}, which is fully robust against weak instruments \citep{staiger_instrumental_1997}. We show that our method controls the conditional Type I error of a test $T$ accounting for selection and provides conditional confidence intervals for the treatment effect that takes into account instrument selection; the conditional Type I error is also referred to as \emph{selective Type I error} \citep{fithian_optimal_inference_2014}. Also, as long as sisVIVE selects instruments which contain all invalid instruments with high probability, our conditional method also controls the usual marginal Type I error and the conditional confidence interval has the usual marginal coverage; here, marginal Type I error and marginal coverage refers to the usual inferential properties that do not condition on a set of instruments. We also extend our sampling approach to work with summary-level data, a popular data format in MR where only summary statistics from GWAS are available to study the exposure effect \citep{burgess2013mendelian,pierce2013efficient}.


We recognize that in practice, MR investigator may use a combination of data-driven selection algorithms \citep{andrews_consistent_1999,andrews_consistent_2001,windmeijer2016use,guo2016causalci} including sisVIVE and/or subjective knowledge to select the final $s$ instruments that are plausibly valid. We focus on sisVIVE because unlike subjective selection of instruments, sisVIVE formalizes the instrument selection process and as such, is tractable to quantitatively characterize the selection effects on inference; in contrast, the effect of selection based on subject-matter knowledge is difficult to formalize since it is investigator-dependent and different investigators may use a wide variety of methods for instruments selection. Also, some instrument selection methods are computationally expensive \citep{andrews_consistent_1999,andrews_consistent_2001}. \citet{bi_inferactive_2017} describes a general interactive data analysis framework that might encompass these possibilities.

Despite the caveats, we believe our method is a more honest alternative than the traditional practice of simply ignoring how instrument selection occurred. By studying a particular example of instrument selection, we hope to bring attention to this problem to the MR community. We also point out ways to extend our framework beyond sisVIVE. 
In particular, our framework is applicable as long as the instrument selection process is expressed in terms of a convex program, and the test statistic for the exposure is asymptotically Normal.

\section{Instrumental Variables Model} \label{sec:model}

\subsection{Notation}
For each individual $i=1,\ldots,n$, let $Y_i \in \reals$ be the outcome, $D_i \in \reals$ be the treatment/endogenous variable, and $Z_{i} \in \reals^L$ be the $L$ candidate instruments. Let $Y = (Y_1,\ldots,Y_n)$, $D = (D_1,\ldots,D_n)$, and $Z \in \reals^{n \times L}$ of instruments. For any full rank matrix $Z \in \reals^{n \times L}$, let $P_{Z} = Z(Z^\intercal Z)^{-1} Z^\intercal$ be the orthogonal projection matrix onto the columns of $Z$ and $P_{Z^\perp} = I - P_{Z}$ be the residual orthogonal projection matrix. For any real vector $v$, let ${\rm sign}(v)$ denote the vector of the signs of $v$. For any set $E \subseteq \{1,\ldots,L\}$, let $|E|$ be the cardinality of the set $E$ and $E^C$ be the complement. Let $Z_{E}$ and $Z_{-E}$ be the  $n$ by $|E|$ and $n$ by $L - |E|$ matrices, respectively, where the columns consist of the elements in the sets $E$ and $\{1,\ldots,L\} \setminus E$, respectively. Let ${\rm det}(\cdot)$ denote the determinant of a matrix. Let ${\mathbb{I}(\cdot)}$ be the indicator function. Let $\textbf{1}$ denote an all-one vector of dimension implied from context. We adopt the usual big-O and small-O notations, $O(\cdot)$ and $o(\cdot)$, to denote the order of a function.

\subsection{Review: Model and Definition of IV} \label{sec:def}
We follow the MR literature, specifically \citet{burgess_robust_2016} and \citet{bowden_framework_2017}, and consider an additive, linear, constant effects structural model that governs the distribution of the observed outcome, exposure, and instruments $(Y_i,D_i,Z_i)$: 
 \begin{equation}
 \begin{aligned} \label{eq:model}
Y_i &=  Z_{i}^\intercal \alpha^* +  D_i \beta^* + \delta_i \\
D_i &=  Z_{i}^\intercal \gamma^* + \xi_{i2}  \\
(\delta_i, \xi_{i2} \mid Z_{i}) &\iid H(0, \Sigma^*), \quad{} Z_{i} \iid F 
\end{aligned}
\end{equation}
The unknown finite-dimensional parameters in the model are $\alpha^*$, $\beta^*$, $\gamma^*$, and $\Sigma^*$. The instruments $Z_{i.}$ for individual $i$ is generated from a distribution $F$ where $E(Z_{i})$ and $E(Z_{i} Z_{i}^\intercal)$ exist and $E(Z_{i} Z_{i}^\intercal)$ is positive definite. The distribution of the errors $H$ is any bivariate distribution with mean zero and covariance $\Sigma^*$ that does not depend on $Z_{i}$; typically, $H$ is assumed to be bivariate Normal. Also, without loss of generality, we assume that $Y$, $D$, and $Z$ have been centered to mean zero, which allows us to remove the intercept terms in equation \eqref{eq:model}. We can also incorporate exogenous covariates $X$ as a linear additive term to model \eqref{eq:model} and project $X$ out to arrive back at the structural model in \eqref{eq:model}; see chapters 2.4 and 2.5 of \citet{davidson_estimation_1993} or equation 2.3 of \citet{andrews_optimal_2006} for examples.

The target parameter for inference is $\beta^*$, which represents the effect of the exposure on the outcome. Parameter $\alpha^*$ in model \eqref{eq:model} represents a combination of the instruments' direct effect on the outcome as well as the confounding effect due to unmeasured confounders of the exposure-outcome relationship; see 
Section 2 of \citet{bowden_framework_2017} for detailed interpretations of these parameters. If all instruments are valid and therefore satisfy (A2) and (A3), we have $\alpha^* = 0$ and equation \eqref{eq:model} would reduce to the usual instrumental variables model. 
If some instruments are invalid, $\alpha^* \neq 0$ and (A2) or (A3) would be violated. In short, the parameter $\alpha^*$ can be used to define assumptions (A2) and (A3), which we state below.
\begin{definition} \label{def:pleiotropicIV}
Suppose we have $L$ candidate instruments along with model \eqref{eq:model}. We say that instrument $j = 1,\ldots,L$ satisfies both (A2) and (A3), or is valid, if $\alpha_j^* = 0$ and is invalid if $\alpha_j^* \neq 0$. 
\end{definition}
\noindent Definition \ref{def:pleiotropicIV} is identical to MR's definition of invalid and valid instruments \citep{bowden_consistent_2016,burgess_robust_2016}. Definition \ref{def:pleiotropicIV} also matches the definition of a valid instrument in \citet{holland_causal_1988} and \citet{angrist_identification_1996} if there is a single instrument (i.e $L = 1$). When there are several instruments (i.e. $L > 1$), Definition \ref{def:pleiotropicIV} can be viewed as a generalization of the definitions in \citet{holland_causal_1988} and \citet{angrist_identification_1996}. 

Parameter $\gamma^*$ in model \eqref{eq:model} represents the instruments' associations to the exposure and we can use its support to formalize (A1) as follows. 
\begin{definition} \label{def:strongIV}
Suppose we have $L$ candidate instruments along with model \eqref{eq:model}. We say that instrument $j = 1,\ldots,L$ satisfies (A1), or is a non-redundant IV, if $\gamma_j^* \neq 0$.
\end{definition}
\noindent In MR, all instruments are generally assumed to be marginally associated with the treatment, i.e. $\gamma_j^* \neq 0$ for all $j$. Also, when there is only one candidate instrument (i.e. $L =1$), Definition \ref{def:strongIV} of (A1) is a special case of the more general definition of (A1) in \citet{angrist_identification_1996}. 

\subsection{Review: Point Estimators and Test Statistics} \label{sec:point_est}
We review point estimators of the model parameters and test statistics for $\beta^*$. Additional details can be found in the supplementary materials where we enumerate the exact technical conditions and derive their statistical properties.

Consider any set $E \subseteq \{1,\ldots,L\}$ of the plausibly invalid instruments where $E^* \subseteq E$ and $\gamma_{E^C}^* \neq 0$. The most popular and consistent estimator for $\beta^*$ is the two stage least squares (TSLS) estimator
\begin{equation} \label{eq:tsls_est} 
\widehat{\beta}_{\rm TSLS} = \frac{D^\intercal (P_Z - P_{Z_E})Y}{D^\intercal (P_Z - P_{Z_E}) D}
\end{equation}
For inference of $\beta^*$, we need a consistent estimator of the covariance $\Sigma^*$ in model \eqref{eq:model}. One such estimator is by plugging in the estimator $\widehat{\beta}_{\rm TSLS}$ into model \eqref{eq:model} and residualizing
\begin{align*}
\widehat{\Sigma}(\widehat{\beta}_{\rm TSLS}) &=  \frac{1}{n} \begin{pmatrix} Y - D\widehat{\beta}_{\rm TSLS}  & D \end{pmatrix}^\intercal  P_{Z^\perp} \begin{pmatrix} Y - D\widehat{\beta}_{\rm TSLS} & D \end{pmatrix},\quad{}
\begin{pmatrix} Y - D\widehat{\beta}_{\rm TSLS} & D \end{pmatrix} \in \reals^{n \times 2}
\end{align*}
Alternatively,  we can construct a consistent estimator of $\Sigma^*$ under the null hypothesis $H_0: \beta^* = \beta_0$ by replacing $\widehat{\beta}_{\rm TSLS}$ with the null value $\beta_0$:
\begin{equation} \label{eq:sigma_est} 
\widehat{\Sigma}(\beta_0) =  \frac{1}{n} \begin{pmatrix} Y - D\beta_0  & D \end{pmatrix}^\intercal  P_{Z^\perp} \begin{pmatrix} Y - D\beta_0 & D \end{pmatrix}
\end{equation}

For testing the null hypothesis $H_0: \beta^* = \beta_0$, we review two popular tests. The first is based on the two-stage least squares (TSLS) and is defined as
\begin{equation} \label{eq:tsls_test}
T_{\rm TSLS} = \frac{D^\intercal (P_Z - P_{Z_E}) (Y - D\beta_0) }{\sqrt{\widehat{\Sigma}_{11}} \sqrt{D^\intercal (P_Z - P_{Z_E}) D}}
\end{equation}
where $\widehat{\Sigma}_{11}$ is a consistent estimator of $\Sigma_{11}^*$, a component of $\Sigma^*$. Under standard asymptotic arguments (c.f. Section 5.2.2 of \cite{wooldridge_econometrics_2010}), $T_{\rm TSLS}$ converges to a standard Normal under $H_0$. The second test statistic is the Anderson-Rubin (AR) test \citep{anderson_estimation_1949} and is defined as
 \begin{align} \label{eq:AR_test}
T_{\rm AR} &=\frac{(Y - D\beta_0)^\intercal (P_Z - P_{Z_E}) (Y - D\beta_0) / (L - |E|)}{(Y - D\beta_0)^\intercal P_{Z^\perp} (Y - D\beta_0)/(n - L)}
\end{align}
Under $H_0$ and if the errors in model \eqref{eq:model} are a bivariate Normal, $T_{\rm AR}$ follows an F distribution with degrees of freedom $L - |E|$ and $n - L$. A unique feature about the Anderson-Rubin test is that if the instruments are weak, it still provides valid Type I error control whereas the TSLS test does not \citep{staiger_instrumental_1997}.

\subsection{Review: Instrument Selection with sisVIVE}
\citet{kang_instrumental_2016} proposed a data-driven method to select valid instruments from a candidate set of $L$ instruments under model \eqref{eq:model}. The algorithm, called sisVIVE, is a modification of the Lasso where we take model \eqref{eq:model} and penalize the parameter $\alpha$, whose support defines valid and invalid instruments in Definition \ref{def:pleiotropicIV}.
\begin{equation} \label{eq:sisvive}
\widehat{\alpha}, \widehat{\beta} \in \amin{\alpha,\beta} \frac{1}{2} ||P_Z (Y - Z\alpha - D\beta)||_2^2 + \lambda ||\alpha||_1, \quad 0 < \lambda
\end{equation}
Under some conditions, \citet{kang_instrumental_2016} showed that the estimate $\widehat{\beta}$ is consistent for $\beta^*$. \citet{windmeijer2016use} studied the instrument selection property by analyzing the support of $\hat{\alpha}$ and provided conditions where sisVIVE consistently selects valid instruments. sisVIVE has been used by MR investigators, such as \citet{allard_mendelian_2015} and \citet{bao2019assessing}, who selected instruments based on sisVIVE and conducted further analysis of $\beta^*$ with the selected instruments.
Unfortunately, both papers did not adjust their test of $\beta^*$ to account for instrument selection from sisVIVE. Our goal is to address the inferential issues after instrument selection as formalized in the next section.

\subsection{Statistical Goal: Pivotal Conditional Inference} \label{sec:goal}
Let $E = \rm supp (\widehat{\alpha})$ be the set of selected invalid instruments from sisVIVE. In MR, the traditional inferential goals are to test the null hypothesis of the exposure effect $H_0: \beta^* = \beta_0$ for some $\beta_0$ where the Type I error is under a pre-specified $\pi, 0 < \pi < 1,$ and to construct a corresponding $1-\pi$ confidence interval for $\beta^*$. Typically, this is done by using a test statistic $T$ and deriving a pivotal null distribution of $T$ that does not depend on unknown parameters. 
\begin{equation} \label{eq:null_marg}
P_{H_0: \beta^* = \beta_0} (T \geq t), \quad{} t \in \reals 
\end{equation}
However, as mentioned in Section \ref{sec:intro_condinf}, the null distribution in equation \eqref{eq:null_marg} fails to recognize two things. First, the selected set of invalid instruments $E$ is random because it is a function of the data $Y$, $D$, and $Z$, which in turn was randomly sampled. But, the traditional derivation of the null distribution assumes $E$ to be pre-specified and fixed. Second, the null distribution does not reflect that instrument selection has taken place. In fact, a statistical quantity that addresses these two concerns and consequently, is a better reflection of a typical MR analysis would be a conditional null
\[
P_{H_0: \beta^* = \beta_0} (T \geq t\mid E \text{ was selected}) 
\]
For sisVIVE, a variant of the above is a finer conditioning event
\begin{equation} \label{eq:null_cond}
P_{H_0: \beta^* = \beta_0} (T \geq t \mid{\rm supp}(\alpha) = E, {\rm sign}(\alpha_{E}) = \widehat{s}_E), \quad{} t \in \reals 
\end{equation}
Here, $\alpha$ is the optimization variable in \eqref{eq:sisvive}, $E$ is the observed selected set of invalid instruments and $\widehat{s}_{E}$ is the observed sign of $\widehat{\alpha}_E$, the latter two from sisVIVE \eqref{eq:sisvive}. The condition ${\rm supp}(\alpha) = E$ represents selecting $E$ instruments from sisVIVE and consequently, restricting the support of $\alpha$ in the model. The condition ${\rm sign}(\alpha_{E}) = \widehat{s}_E$ represents observing the estimated signs of the direct effects from selected instruments. For example, among a pool of $L=100$ instruments, if the first ten were selected from sisVIVE, $E = \{1,\ldots,10\}$ and  ${\rm sign}(\alpha_{E})$ would be set to the signs of the first 10 instruments estimated from sisVIVE. Unlike the marginal null in equation \eqref{eq:null_marg}, the conditional null in equation \eqref{eq:null_cond} acknowledges that the invalid instruments $E$ were selected from a pool of $L$ instruments, in this case with sisVIVE, and we observed which direction the invalid instruments affected the outcome before testing the exposure effect $\beta^*$.  
Practically speaking, equation \eqref{eq:null_cond} resembles the MR analysis in \citet{allard_mendelian_2015} where the authors tested $H_0$ after selecting instruments using sisVIVE. and is a more true-to-life representation of an MR analysis for the exposure effect than the traditional null distribution \eqref{eq:null_marg} that ignores instrument selection. Also, while we do not explore it in this paper, one can also study a modified version of \eqref{eq:null_cond} without conditioning on the signs by taking unions of different sign patterns and integrating it out from \eqref{eq:null_cond}; see \citet{lee_exact_lasso_2016} for details.

If the conditional distribution in equation \eqref{eq:null_cond} is pivotal, it allows us to construct tests and confidence intervals that control the conditional Type I error at level $\pi$ 
\[
P_{H_0: \beta^* = \beta_0} (\text{reject } H_0 \mid{\rm supp}(\alpha) = E,  {\rm sign}(\alpha_{E}) = \widehat{s}_E) \leq \pi
\]
and conditional coverage of $\beta^*$ at $1-\pi$ \citep{fithian_optimal_inference_2014,bi_inferactive_2017}. The conditional Type I error, unlike the traditional Type I error, i.e.  $P_{H_0: \beta^* = \beta_0} (\text{reject } H_0)$, takes into account instrument selection that was done a priori. Also, as long as $E^C$ only contains valid instruments, any test that controls the conditional type I error at level $\pi$ and achieves nominal conditional coverage will control the usual Type I error at level $\pi$ and achieve the usual $1-\pi$ coverage. 

Since the conditional null distribution plays a key role in achieving our inferential goals, the rest of the paper is devoted to charactering it. 

\subsection{Why Not Sample Splitting?} \label{sec:sample_split}
One solution to deriving the conditional distribution in equation \eqref{eq:null_cond} is sample splitting, where a subset of the sample is used to select the instruments and the rest of the unused samples is used to compute the test statistic; this is also the approach taken by \citet{burgess2011avoiding, zhao2018powerful}. Under this sampling scheme, the conditional null simply becomes the marginal null and the traditional inference will be honest. However, a major shortcoming with data splitting is that there is a loss in sample size in both testing and selecting instruments; we are using fewer samples to select good instruments and to test for exposure effect. Indeed, if the effect size is small, the loss of power from sample splitting would make it difficult to discover such effects. \citet{fithian_optimal_inference_2014} also proved under general conditions that sample splitting is inadmissible to another procedure called {\emph{data carving}}, which uses entire dataset for inference after conditioning on the selection result. They also showed that data carving with holdout significantly boosts power of a test compared to no holdout. Our method described in Section \ref{sec:method} also uses a type of data carving with holdout where we randomize the instrument selection algorithm.

Also, in practice, MR investigators use the same sample to assess the validity of their selected instruments; see \citet{voight_plasma_2012} and \citet{allard_mendelian_2015} for examples. This is partly because to achieve full independence between the selection event and the computation of the test statistic, MR investigators have to find three independent GWAS that measured the exposure and the outcome, which may be difficult for non-common exposures or outcomes. 

\section{Sampling Method for the Pivotal Conditional Null Distribution} \label{sec:method}
We characterize the conditional null distribution in equation \eqref{eq:null_cond} using a hit-and-run sampling procedure. The sampling procedure is roughly divided into three parts: the randomized selection algorithm, the reparametrization of \eqref{eq:null_cond} to simplify the conditioning event $E$, and a standard hit-and-run or MCMC sampling algorithm. To simplify discussion, we begin by initially assuming (i) $Z$ is fixed and (ii) the nuisance parameters  $\alpha^*$, $\gamma^*$, and $\Sigma^*$ are known. Section \ref{sec:asymptotic} removes the two restrictions by deriving an asymptotic version of the conditional null distribution.

\subsection{Randomized Instrument Selection} \label{sec:exact_rand}
The first step in our method is to recast sisVIVE \eqref{eq:sisvive} as a randomized instrument selection algorithm
\begin{equation} \label{eq:rsisvive}
(\widehat{\alpha},\widehat{\beta}) = \amin{\alpha,\beta} \frac{1}{2} \| P_Z(Y - Z\alpha - D\beta)\|_2^2 + \lambda ||\alpha||_1 - \omega^\intercal \cdot
\left(\begin{array}{l}
\beta\\
\alpha\\
\end{array}
\right)
 + \frac{\epsilon}{2} 
 \left\| \left(\begin{array}{l}
 \beta \\
 \alpha \\
 \end{array}
\right) \right\|_2^2
\end{equation}
The randomization term $\omega$ comes from a density $g(\omega)$ that is specified by the investigator and is independent of $Y$ and $D$, typically Gaussian, Laplacian or other heavy tailed distributions with variance on the order of $O(n)$. 
The quadratic term ensures strong convexity and existence of a solution. The penalty term $\lambda$ and $\epsilon$ are chosen so that its order is $O(n^{\frac{1}{2}})$ and $o(1)$, respectively; see Section \ref{sec:simulation} for examples of these values.

The idea of randomizing the original selection algorithm to derive conditional nulls was proposed by \citet{tian_randomized_2016}. The authors showed that a randomized algorithm is mathematically equivalent to data carving with holdout. More generally, randomized selection algorithms belong to a family of resampling algorithms, such as sample splitting and cross validation. However, the randomized selection algorithm in \eqref{eq:rsisvive} has some advantages, especially compared to sample splitting. First, as mentioned in Section \ref{sec:sample_split}, using \eqref{eq:rsisvive} increases the power to test the exposure effect. Second, the randomized algorithm allows any pre-existing test statistics for the exposure effect to be used for downstream analysis; we do not have to create a new test statistic so that the conditional Type I error is controlled. Finally, as we'll see below, the conditional density under \eqref{eq:rsisvive} becomes tractable for off-the-shelf hit-and-run and MCMC algorithms. 


\subsection{Exact Conditional Density} \label{sec:exact}
Under a fixed $Z$, let $\ell (Y, D, \omega \mid {\rm supp}(\alpha) = E, {\rm sign}(\alpha_{E}) = \widehat{s}_E )$ denote the density of the data $Y, D$ and the known independent randomization term $\omega$ conditional on the same selection events. Let $\ell (Y, D \mid {\rm supp}(\alpha) = E, \text{sign}(\alpha_{E}) = \widehat{s}_E)$ denote the density of the data $(Y,D)$ that marginalized out $\omega$ from the first density. Given samples of $Y$ and $D$ from the conditional density, the conditional density of the test statistic $T(Y, D)$ is simply a plug-in of $Y$ and $D$ into $T(Y,D)$. In short, to obtain a conditional distribution of $T(Y,D)$, we need to characterize $\ell (Y, D, \omega \mid {\rm supp}(\alpha) = E, {\rm sign}(\alpha_{E}) = \widehat{s}_E )$.

Directly sampling the conditional density of $(Y, D, \omega)$  requires solving sisVIVE until the conditioning events $\{ {\rm supp}(\alpha) = E, {\rm sign}(\alpha_{E}) = \widehat{s}_E \}$ are met, which can be a computationally expensive task. Importantly, the restrictions that the conditioning events impose on the data $Y$ and $D$ are non-trivial, making it difficult to analytically marginalize out $\omega$. But, by reparametrizing $\ell(Y, D, \omega \mid {\rm supp}(\alpha) = E, {\rm sign}(\alpha_{E}) = \widehat{s}_E )$ with a familiar change-of-variables formula, we can alleviate these two concerns. Specifically, the reparametrization is based on the convexity in equation \eqref{eq:rsisvive}, which implies that the solutions $\alpha,\beta$ to the optimization must satisfy the Karush-Kuhn-Tucker (KKT) condition \citep{boyd_convex_2004}
\[
- \left(\begin{array}{l}
D^\intercal \\
Z^\intercal \\
\end{array}\right) P_Z (Y - Z \alpha - D\beta) + \lambda 
\left(\begin{array}{l}
u\\
0\\
\end{array}\right) - \omega + \epsilon \left(\begin{array}{l}
\alpha\\
\beta\\
\end{array}
\right)
= 0, \ \rm sign (\alpha_E) = u_E = \widehat{s}_{E},\ \alpha_{-E} = 0, \ \|u_{-E}\|_\infty \leq 1
\]
where $\widehat{s}_E$ is the observed signs for the set $E$, $\alpha, \beta$ are the optimization variables and $u$ is the subgradient. We remind readers that we use $\alpha, \beta, u$ to denote the optimization variables of \eqref{eq:rsisvive} and $\widehat{\alpha}, \widehat{\beta}, \widehat{u}$ to denote a specific solution of the optimization in \eqref{eq:rsisvive} from the dataset. Importantly, the KKT condition provides a mapping between the density of $(Y,D,\omega)$ and the density based on the optimization variables $(Y, D, \beta,\alpha_E,u_{-E})$ by using a change-of-variable formula. 
\begin{theorem}[Exact conditional density via reparametrization] \label{thm:exact} 
The conditional density of $Y, D, \omega$ can be expressed (up to a proportionality constant) with respect to the variables $Y, D, \beta,\alpha,u$, i.e. 
\begin{equation} \label{eq:density}
\begin{split}
&\ell (Y,D,\omega \mid {\rm supp}(\alpha) = E, {\rm sign}(\alpha_{E}) = \widehat{s}_E )  \\
\propto&
f (Y, D)\cdot g\left(- \left(\begin{array}{l}
D^\intercal \\
Z^\intercal \\
\end{array}\right) P_Z (Y - Z \alpha - D\beta)  + \lambda 
\left(\begin{array}{l}
u\\
0\\
\end{array}\right) + \epsilon \left(\begin{array}{l}
\alpha\\
\beta\\
\end{array}
\right)  \right) \cdot |\mathcal{J}| \cdot \mathbb{I}(\mathcal{B})
\end{split}
\end{equation}
where $\mathcal{B} = \{ {\rm sign}(\alpha_{E}) = \widehat{s}_{E}, \alpha_{-E} = 0, u_E = \widehat{s}_E, \|u_{-E}\|_\infty \leq 1 \}$ and
\begin{align*}
|\mathcal{J}| &= \det \begin{pmatrix} \begin{pmatrix}D^\intercal P_Z D  & D^\intercal Z_E  \\
Z_E^\intercal D & Z_E^\intercal Z_E \end{pmatrix} + \epsilon I \end{pmatrix} \cdot \lambda^{|-E|}
\end{align*}
\end{theorem}
In addition to the reparametrization, Theorem \ref{thm:exact} shows the effect that conditioning on instrument selection has on the original data density $Y,D$, represented by $f(Y,D)$ in \eqref{eq:density}. The terms on the right hand side of $f(Y,D)$ essentially reweigh and constrain the distribution of $Y$ and $D$ to reflect that selecting the instruments from \eqref{eq:rsisvive} changed the original data density. The constraints on $Y$ and $D$ are expressed by optimization variables where the original conditioning event, $\{ {\rm supp}(\alpha) = E, \text{sign}(\alpha_{E}) = \widehat{s}_E \}$, is re-expressed as the event $\mathcal{B}$. $\mathcal{B}$ is special in that it is a set of simple quadrant and box constraints on said optimization variables compared to the original conditioning event. Now, if the original data density $f(Y, D)$ is known, i.e. if we know the model parameters $\alpha^*, \gamma^*, \Sigma^*$, the density is pivotal with respect to these parameters and one can directly use \eqref{eq:density} in an MCMC sampler to generate samples of $Y$ and $D$, which can then be used to compute the conditional null distribution of any test statistic for the exposure effect. However, in practice, these model parameters, i.e. $f(Y,D)$, are unknown and act as nuisance parameters. The next section remedies this issue by conditioning on an asymptotic sufficient statistics of the nuisance parameters in $f(Y,D)$ and deriving an asymptotic conditional density of the test statistic $T$.

\subsection{Asymptotic pivotal conditional densities of TSLS and AR test statistics} \label{sec:asymptotic}
In this section, we present the conditional density of the TSLS and the AR test statistics in Section \ref{sec:point_est}. The key technical step is to apply the selective central limit theorem (CLT) which was proposed in \citet{markovic_bootstrap_2017}. We state the results of applying the theorem in the main paper and the supplement contains additional details. 
Throughout this section, $Y,D,Z$ and functions thereof, such as $\widehat{T}$, $\widehat{\beta}$, $\widehat{\alpha}$, and $\widehat{u}$ are observed values (i.e. fixed) and the sampling variables (i.e. those used in the sampling procedure to generate conditional null distributions) are $T$, $\beta$, $\alpha$, and $u$. 

For the TSLS test statistic $T_{\rm TSLS}$ in equation \eqref{eq:tsls_test}, Section \ref{sec:point_est} showed that it asymptotically follows a standard Normal with mean $\mu_T = 0$ and variance $\Sigma_{T} = 1$, which we denote as $\phi_{(0, \Sigma_T )} (T_{\rm TSLS})$. By utilizing the selective CLT, the asymptotic conditional density of $T_{\rm TSLS}$ under the null hypothesis $H_0: \beta^* = \beta_0$ is the reweighing of the usual density of $\phi_{(0, \Sigma_T )} (T_{\rm TSLS})$:
\begin{equation}  \label{eq:tsls_cden} 
\begin{split}
&\ell_{\beta_0}(T_{\rm TSLS}, \omega \mid {\rm supp}(\alpha) = E, \text{sign}(\alpha_{E}) = \widehat{s}_E, F) \\
&\propto \phi_{(0, \Sigma_T )} (T_{\rm TSLS})\cdot g\left\{- \Sigma_{S,T}\Sigma_T ^{-1} \cdot T_{\rm TSLS} + \left[\begin{pmatrix}Z^\intercal Z & Z^\intercal D \\ D^\intercal Z & D^\intercal P_{Z} D\end{pmatrix} + \epsilon I \right] \cdot \begin{pmatrix}\alpha \\ \beta \end{pmatrix} + \lambda u - F \right\} \cdot |\mathcal{J}| \cdot \mathbb{I}(\mathcal{B})
\end{split}
\end{equation}
where
\begin{align*}
\Sigma_T &= 1, \quad{} \Sigma_{S,T} =  \sqrt{\frac{\widehat{\Sigma}_{11}}{D^\intercal (P_Z - P_{Z_E}) D}} \begin{pmatrix} Z^\intercal (P_Z - P_{Z_E}) D \\ D^\intercal (P_Z - P_{Z_E}) D\end{pmatrix} \\
F &= \begin{pmatrix} Z^T Y \\ D^T P_Z Y \end{pmatrix} - \Sigma_{S,T} \Sigma^{-1}_T \cdot T
\end{align*}
and $T$ is the test statistic. The term $F$ is, in the asymptotic sense, a sufficient statistic for the unknown parameters of the data distribution $Y,D$ and conditioning on it frees us from these nuisance parameters, akin to conditioning on sufficient statistics in a parametric bootstrap for linear regression; see \citet{markovic_bootstrap_2017} and the supplementary materials for a detailed explanation.

Similarly, let $\phi_{\left(\beta_0, \Sigma_T \right)} (\beta_{\rm TSLS})$ denote the density of the TSLS estimator $\beta_{\rm TSLS}$ in equation \eqref{eq:tsls_est}. We can also derive the conditional distribution of the TSLS estimator as a direct application of the selective CLT.
\begin{equation}
\begin{split}
&\ell_{\beta_0} (\beta_{\rm TSLS}, \omega \mid {\rm supp}(\alpha) = E, \text{sign}(\alpha_{E}) = \widehat{s}_E, F) \\
&\propto \phi_{\left(\beta_0, \Sigma_T \right)} (\beta_{\rm TSLS}) \cdot g\left\{-\Sigma_{S,T}\Sigma_T^{-1} \cdot \beta_{\rm TSLS} + \begin{pmatrix}\begin{pmatrix}Z^\intercal Z & Z^\intercal D \\ D^\intercal Z & D^\intercal P_Z D\end{pmatrix} + \epsilon I\end{pmatrix}\begin{pmatrix}\alpha \\ \beta\end{pmatrix} + \lambda u - F\right\}\cdot |\mathcal{J}| \cdot \mathbb{I}(\mathcal{B})
\end{split}
\end{equation}
where $F$ is identical as before except for the test statistic and the variances and
$$
\Sigma_T = \frac{\widehat{\Sigma}_{11}}{D^\intercal (P_Z - P_{Z_E}) D},\quad{}\Sigma_{S,T} = \frac{\widehat{\Sigma}_{11}}{D^\intercal (P_Z - P_{Z_E}) D}\begin{pmatrix}Z^\intercal (P_Z - P_{Z_E}) D \\ D^\intercal (P_Z - P_{Z_E}) D\end{pmatrix}
$$

For the Anderson-Rubin test statistic $T_{\rm AR}$, let $\widetilde{T} \in \reals^L$ be an intermediary sampling target
$$
\widetilde{T} = Z^\intercal (I - P_{Z_E}) (Y - D\beta_0)
$$
and $\phi_{\left[0, \Sigma_T \right]} (\widetilde{T})$ be the unconditional density of $\widetilde{T}$. Then, the conditional sampling density of $\widetilde{T}$ under the null $H_0: \beta^* = \beta_0$ is given by:
\begin{equation} \label{eq:ar_cden}
\begin{split}
&\ell_{\beta_0} (\widetilde{T}, \omega \mid {\rm supp}(\alpha) = E, \text{sign}(\alpha_{E}) = \widehat{s}_E, F) \\
&\propto \phi_{\left[0, \Sigma_T \right]} (\widetilde{T}) \cdot g\left\{- \Sigma_{S, T} \Sigma_{T}^{-1} \widetilde{T} +\left[ \begin{pmatrix}Z^\intercal Z & Z^\intercal D \\ D^\intercal Z & D^\intercal P_Z D \end{pmatrix} + \epsilon \cdot I \right]\begin{pmatrix}\alpha \\ \beta \end{pmatrix}  + \lambda u - F\right\} \cdot |\mathcal{J}| \cdot \mathbb{I}(\mathcal{B})
\end{split}
\end{equation}
where $F$ is identical as before except for the test statistic and the variances and
\begin{align*}
\Sigma_T &= \widehat{\Sigma}_{11} \cdot Z^\intercal (I - P_{Z_E}) Z, \quad{} \Sigma_{S,T} = \widehat{\Sigma}_{11} \cdot \begin{pmatrix}Z^\intercal (I - P_{Z_E}) Z \\ D^\intercal (I - P_{Z_E}) Z\end{pmatrix}
\end{align*}
We can then generate samples of the Anderson-Rubin test $T_{\rm AR}$ by sampling $\widetilde{T}$ from \eqref{eq:ar_cden} and plugging into 
\begin{equation} \label{eq:ar_tilde}
T_{\rm AR} = \frac{\widetilde{T}^\intercal \cdot (Z^\intercal Z)^{-1} \cdot \widetilde{T} / (L-|E|)}{(Y - D\beta_0)^\intercal P_{Z^\perp} (Y - D\beta_0)/(n - L)}
\end{equation}
to generate conditional null distributions of  $T_{\rm AR}$ under instrument selection. 

\subsection{Sampling Algorithm} \label{sec:sampling}
Given the conditional density of a test statistic, we can use any off-the-shelf MCMC sampling method to sample values of the test statistic. We detail one Gibbs sampling algorithm and characterize the conditional density of the TSLS test statistic in \eqref{eq:tsls_cden}.

Let $K_1, K_2$, and $K_3$ be the following quantities.
\[
K_1 = -\Sigma_{S,T}, \ K_2 = \left[\begin{pmatrix}Z^\intercal Z & Z^\intercal D \\ D^\intercal Z & D^\intercal P_{Z} D\end{pmatrix} + \epsilon \cdot I \right], \ K_3 = \left[\begin{pmatrix}Z^\intercal Y \\ D^\intercal P_Z Y\end{pmatrix} - \Sigma_{S,T}\cdot \widehat{T}_{\rm TSLS}\right]
\]
These variables are observed values from the data and hence fixed constants in the sampler. Then, the conditional null density in \eqref{eq:tsls_cden} is proportional to 
\begin{align*}
& \ell_{\beta_0}(t, \omega \mid {\rm supp}(\alpha) = E, \text{sign}(\alpha_{E}) = \widehat{s}_E, F) \\
&\propto \phi_{0, 1} (t)\cdot g\left\{- K_1 \cdot t + K_2 \cdot \begin{pmatrix}\alpha_E \\ 0 \\ \beta \end{pmatrix} + \lambda u - K_3 \right\} \cdot \mathbb{I}[(\alpha_E, u_{-E}) \in \mathcal{B}] \\
&= h(t,\beta,\alpha_E,u_{-E})
\end{align*}
Then, as detailed below, we can generate samples from $h(t,\beta,\alpha_E,u_{-E})$ sequentially by using Gibbs.
\begin{enumerate}
\item \textbf{Initialization ($k=0$)}: Initialize $(t, \beta, \alpha_E, u_{-E} )$ and denote them as $(t^{(0)}, \beta^{(0)}, \alpha^{(0)}_{E}, u_{-E}^{(0)} )$
\item \textbf{Gibbs Update}: At step $k$ with variables $(t^{(k)}, \beta^{(k)}, \alpha^{(k)}_{E}, u_{-E}^{(k)} )$
\begin{enumerate}
\item Update $t$ by sampling a proposal, denoted as $\tau^{(k+1)}$, from a proposal distribution that is Normal with mean $\tau^{(k)}$ and variance $a_k^2$;  $a_k$ represents the step size. The proposal distribution is symmetric with an acceptance ratio of 
$$
r = \min \left( \frac{h(t^{(k+1)}, \beta^{(k)}, \alpha_E^{(k)}, u_{-E}^{(k)})}{h(t^{(k)}, \beta^{(k)}, \alpha_E^{(k)}, u_{-E}^{(k)})}, 1 \right)
$$
Accept the proposal $t^{(k+1)}$ with probability $r$. Otherwise, set $t^{(k+1)} = t^{(k)}$.
\item Jointly update $\beta$ and $\alpha_E$ by sampling proposals, denoted as $\beta^{(k+1)}$ and $\alpha_E^{(k+1)}$, via
\begin{align*}
\beta^{(k+1)} &= \beta^{(k)} + b_k \cdot \nu_1, \quad{}& \nu_1 &\sim G \\
\alpha_E^{(k+1)} &= s_E \cdot |\ |\alpha_E^{(k)}| + b_n \cdot \nu_2 |,   \quad{}& \nu_2 &\sim G^{|E|}
\end{align*}
Here, $G$ denotes the distribution of the randomized term $\omega$ in the randomized sisVIVE and $G^{|E|}$ denotes $|E|$ i.i.d. distributions of $G$. The term $b_k$ represents the step size for our update. Both proposals are also symmetric with an acceptance ratio $r$ of
$$
r = \min\left(\frac{g(\omega(t^{(k+1)}, \beta^{(k+1)}, \alpha_E^{(k+1)}, u_{-E}^{(k)}))}{g(\omega(t^{(k+1)}, \beta^{(k)}, \alpha_E^{(k)}, u_{-E}^{(k)}))}, 1 \right)
$$
Accept $\beta^{(k+1)}$ and $\alpha^{(k+1)}_E$ with probability $r$. Otherwise, set $\beta^{(k+1)} = \beta^{(k)}$ and $\alpha^{(k+1)}_E = \alpha^{(k)}_E$.
\item Update $u^{(k+1)}_{-E}$ by sampling a value $v$ from a truncated $G$ distribution at $ \Delta^-, \Delta^+$ 
where
\begin{align*}
\Delta^- &= \Delta -  \lambda \mathbf{1}, \quad{} \Delta^+ =\Delta + \lambda \mathbf{1}, \quad{} \Delta = \left(- K_1 \cdot t^{(k+1)} + K_2 \begin{pmatrix}\alpha_{E}^{(k)} \\ \beta^{(k)} \end{pmatrix} - K_3 \right)_{-E} \\
\end{align*}
and set $u_{-E}^{(k+1)} = \frac{1}{\lambda} (v - \Delta)$.
\end{enumerate}
\item \textbf{Stop} when convergence criterion is reached after 5000 burn-in samples and return 10,000 samples of $t$.
\end{enumerate}
We make a few remarks about the sampling algorithm. First, in practice, we initialize the sampler at the observed values of the sampling variables: $t^{(0)}$ is set to be the observed TSLS test statistic $\widehat{t}_{TSLS}$ with the selected instruments $E$, $\beta^{(0)}, \alpha_E^{(0)}, u_{-E}^{(0)}$ are set to be the solutions from the randomized sisVIVE as $\widehat{\beta}, \widehat{\alpha}_E, \widehat{u}_{-E}$. Second, if $G$ is Laplacian, we typically tune the step sizes $a_k$ and $b_k$ to have acceptance rates around $0.3$. Third, if $G$ is Gaussian, we can use a simpler hit-an-run sampling instead of Gibbs and bypass the requirement to pick step sizes $a_k$ and $b_k$; see \citet{bi_inferactive_2017} for details.

\subsection{Inference via the Conditional Density Accounting for Instrument Selection} \label{sec:goal_solved}
Once we obtain $S$ samples of $t_i, i=1,\ldots,S$ from a sampling algorithm, such as the one described in Section \ref{sec:sampling}, we can achieve the conditional inference goals laid out in Section \ref{sec:goal}. For example, let $\widehat{t}_{TSLS}$ be the observed value of the test statistic from traditional MR analysis after selecting $E^C$ valid instruments from sisVIVE. Instead of comparing this observed test statistic to a marginal null distribution to obtain a p-value, we can now obtain a conditional p-value $p (\beta_0)$ of the null hypothesis of the exposure effect $H_0:\beta^* = \beta_0$ by using the $N$ samples of $t$ generated from the sampler.
\[
p(\beta_0) = \frac{1}{N} \sum_{i=1}^N \mathbb{I} (t_i \geq \widehat{t}_{TSLS})
\]
If the significance level is set to $\pi$, we can reject the null hypothesis if $p(\beta_0) < \pi$. For example, if $\beta_0 = 0$ with (estimated) valid instruments $E^C$ and $p(0) < \pi$, we can say that there is a significant exposure effect \emph{after accounting} for instrument selection done a priori. Also, by the duality between confidence intervals and hypothesis testing, we can construct a $1-\pi$ conditional confidence interval for $\beta^*$ by retaining $\beta_0$'s where $p(\beta_0) \geq \pi$. Because inverting hypothesis for every $\beta_0$ can be computationally burdensome, we use an importance sampling trick from \citet{markovic_bootstrap_2017} to improve computation; see the supplementary materials for implementation details. 

\section{Simulation} \label{sec:simulation}

We conduct a simulation study to demonstrate the proposed method. Formally, we generate data according to model \eqref{eq:model} with the following parameters: $n=1000$, $p=10$, $|E^*| = 3$, $\beta^* = 1$, $\alpha_{j} = 7$ for $j \in E^*$, $\Sigma^*_{11} = \Sigma^*_{22} = 1$, and $\Sigma^*_{12} = 0.8$. The instruments is generated according to standard multivariate Normal, i.e. $Z_i \sim N(0, I_L)$. For the parameters of the procedure \eqref{eq:rsisvive}, we follow \citet{markovic_bootstrap_2017}  and set $\lambda= 2.01 * \sqrt{n \cdot \rm log(n)}$, $\epsilon = 0.01$. The randomization $\omega = (\omega_1,\ldots,\omega_L)$ is taken as having independent Gaussian components with the standard deviation (denoted as $\rm std$) set to be $0.5 * \rm std(P_Z \begin{pmatrix}Z& D\end{pmatrix})*\rm std(P_Z Y)*\sqrt{\frac{n}{n-1}}$.

First, we evaluate whether the null distribution of the conditional p-values is uniform; if the conditional p-values are uniformly distributed, the conditional Type I error is controlled and consequently, the confidence interval would have the desired conditional coverage. We do this for different values of $\gamma^*$ where $\gamma_j^* = r$ for every $j$ and $r$ ranges from $0.04$ to $2.5$. For each $\gamma^*$, we plot the empirical cumulative distribution (CDF) function of the conditional p-values under $H_0: \beta^* = 1$ and compare it with the uniform distribution to see whether the two distributions are identical (i.e. line up along the 45 degree line). Figure \ref{fig:tsls:null} presents the evaluation for the TSLS estimator. For TSLS, except for the case when the instrument is weak and $r \leq 0.06$,  the distribution of the conditional p-values is close to uniform, aligning with what we expect from theory. When the instruments are weak, the empirical distribution deviates from the uniform distribution and our conditional inference suffers from weak instrument bias \citep{staiger_instrumental_1997, stock_survey_2002}. 

Second, we compare our confidence interval to the {\emph{naive}} confidence interval that does not take instrument selection into account, or mathematically treats $E$ as fixed. Specifically, we compare the marginal coverage rates of our interval and the naive interval when $E^C$ contains no invalid instruments. As described in Section \ref{sec:goal}, our $1-\pi$ conditional confidence interval will have $1-\pi$ marginal coverage under this case. In contrast, the naive interval only provides $1-\pi$ coverage if $E$ is fixed and hence, will have lower than $1-\pi$ coverage. Like before, we plot (i) the empirical coverage rates and (ii) the average lengths as a function of $r$ for the proposed conditional interval and the naive confidence interval. Figure \ref{fig:coverage_0_8} presents the evaluation for the test statistic based on the TSLS estimator. For TSLS, while the two intervals are very close, our interval is always closer to the nominal level, set at $95\%$ compared to the naive interval. Our interval also suffers from weak instrument bias as it fails to achieve $95\%$ coverage when $r$ is very small, but not as drastic as the naive interval. Figure \ref{fig:coverage_0_8_p100} presents the same evaluation for TSLS, but when we have more candidate instruments to select from (i.e. $L = 100$) and $|E^*| = 30$; everything else remains the same as before. Here, the effect of instrument selection on coverage is more pronounced between the two intervals, with our interval achieving closer to 95\% coverage than the naive interval. In both cases, the effect from instrument selection on coverage rates for the naive interval decreases as the instruments get stronger, but weak instruments can amplify the selection effect.

We note that we only presented simulation studies where the naive interval has the ``best chance'' of doing well compared to our interval. In the supplementary materials, we conduct more ``idealized'' simulation studies where we design the data generating process so that the difference between our interval and the naive interval is drastic, specifically by varying (i) the magnitude of the instruments' invalidity via $\alpha_E^*$ and $\gamma^*$ associated with invalid instruments and (ii) the magnitude of endogeneity via $\Sigma_{12}^*$. In summary, in finite samples, our TSLS interval can maintain coverage from 80\% to 96\% depending on the instrument strength, invalidity, and endogeneity, whereas the naive interval has coverage rates range from near 0\% to 70\% and never reach the nominal 95\% coverage level.

Also, in the supplementary materials, we repeat the two evaluations under different test statistics to see how our inference procedure is affected by the choice of test statistics. In short, for the AR test, our conditional method shows Type I error control for all values of instrument strength. Also, when comparing coverage of the AR test, we observe that the differences are small, most likely due to the AR test being generally conservative.


\begin{figure}[!ht]
\begin{center}
\centerline{\includegraphics[width=1.\columnwidth]{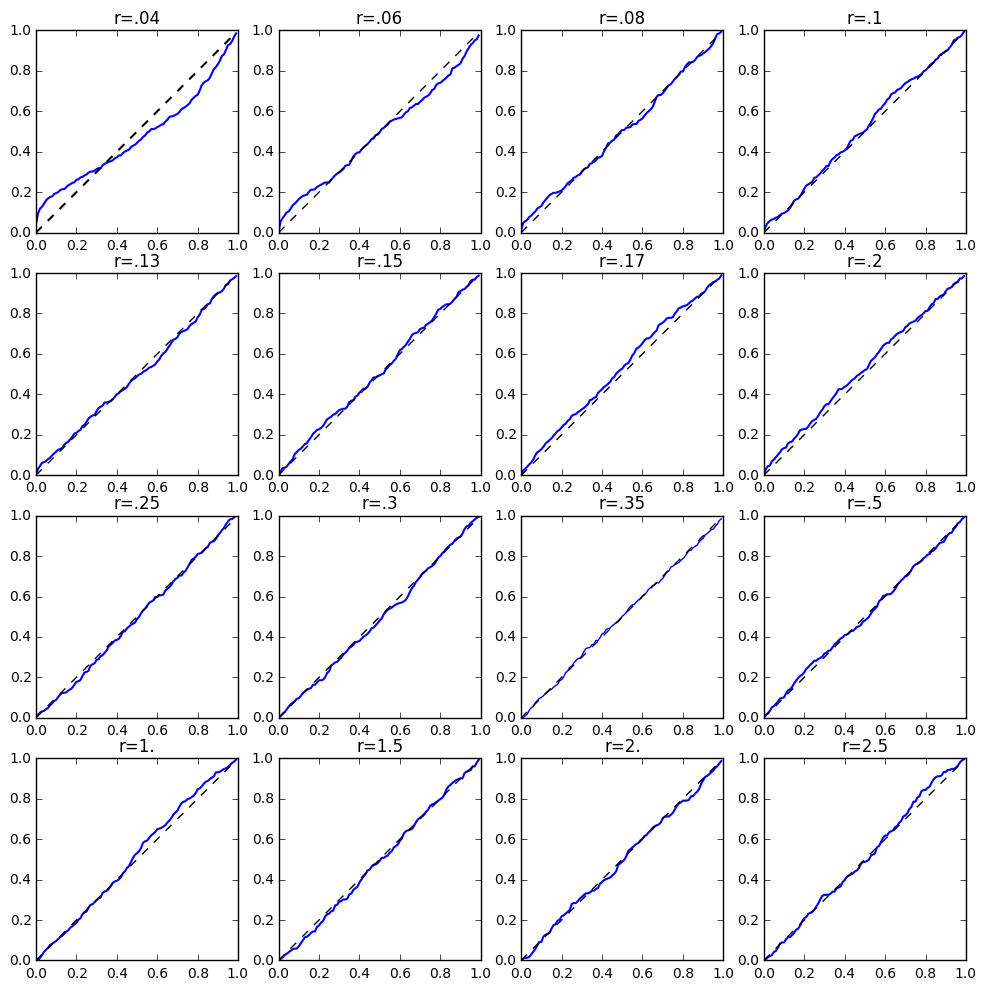}}
\caption{Comparison between the CDF of the empirical distribution of the conditional p-values using TSLS statistic and the uniform distribution. $r$ represents instrument strength as measured by setting $\gamma_j^* = r$ for all $j$.}
\label{fig:tsls:null}
\end{center}
\end{figure}

\begin{figure}[!ht]
\begin{center}
\centerline{\includegraphics[width=.8\columnwidth]{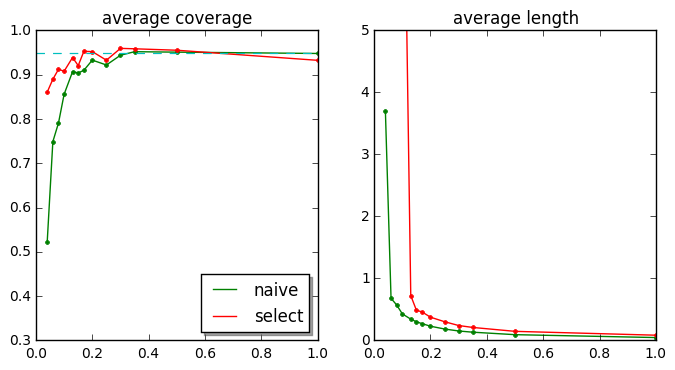}}
\caption{Comparison between the naive and conditional CI using the TSLS test statistic when there are $L = 10$ candidate instruments. The x-axis represents instrument strength as measured by setting $\gamma_j^* = r$. The left plot shows the empirical coverage rate of both intervals and the right plots show the average lengths of both intervals. }
\label{fig:coverage_0_8}
\end{center}
\end{figure}

\begin{figure}[!ht]
\begin{center}
\centerline{\includegraphics[width=.8\columnwidth]{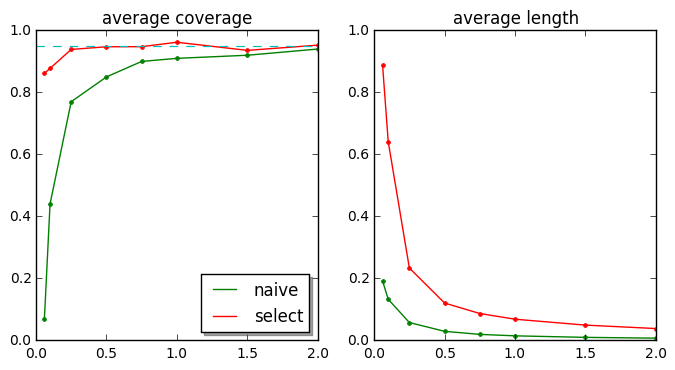}}
\caption{Comparison between the naive and conditional CI using the TSLS test statistic when there are $L = 100$ candidate instruments. The x-axis represents instrument strength as measured by setting $\gamma_j^* = r$. The left plot shows the empirical coverage rate of both intervals and the right plots show the average lengths of both intervals. }
\label{fig:coverage_0_8_p100}
\end{center}
\end{figure} 


\section{Extension: Summary Data in Mendelian Randomization} \label{sec:sum}
While our method was described where individual observations of $Y_i, D_i, Z_i$ were available, majority of MR studies are based on summary statistics due to privacy and logistical reasons \citep{burgess2013mendelian,pierce2013efficient}. The summary statistics are derived from GWAS and they typically consist of estimated regression coefficients that measure the associations of genetic variants with various traits. In this section, we extend our above method to summarized data common in a MR analysis.

Denote $Z_{(j)}$ as the $j$-th column of matrix $Z$. For each instrument $j$, let $\widehat{\beta}_{j,Y}$ be the estimated regression coefficient from running a simple linear regression between $Y$ and instrument $Z_{(j)}$ and let ${\rm se}^2 (\widehat{\beta}_{j,Y})$ be the variance of the estimate $\widehat{\beta}_{j,Y}$. Similarly, for each instrument $j$, let $\widehat{\beta}_{j,D}$ be the estimated regression coefficient from running a simple linear regression between $D$ and $Z_{(j)}$ and let  ${\rm se}^2 (\widehat{\beta}_{j,D})$ be the variance of the estimate $\widehat{\beta}_{j,D}$. Lemma \ref{lem:sumstat} shows that, under suitable assumptions, the statistics $Z^\intercal Y, Z^\intercal D, Z^\intercal Z$ that are critical in the derivation of the conditional density in Section \ref{sec:exact}, can be derived with MR summary statistics,

\begin{lemma} \label{lem:sumstat} Suppose we have the summary statistics $\widehat{\beta}_{j,Y}, {\rm se}^2 (\widehat{\beta}_{j,Y}), \widehat{\beta}_{j,D}, {\rm se}^2 (\widehat{\beta}_{j,D})$ and the data $Y,D,Z$ follows the model in equation \eqref{eq:model}. If 
\begin{itemize}
\item[(i)] $Z^\intercal Z$ is a diagonal, but not necessarily identity, matrix and 
\item[(ii)] $Y, D, Z$ are centered to mean zero,
\end{itemize} 
the quantities $Z^\intercal Y, Z^\intercal D, Z^\intercal Z$ can be written as
\begin{align*}
Z_{(j)}^\intercal Z_{(j)} = \frac{(n-1)\cdot se^2 (\widehat{\beta}_{1,Y}) + \widehat{\beta}_{1,Y}^2}{(n-1)\cdot {\rm se}^2 (\widehat{\beta}_{i,Y}) + \widehat{\beta}_{i,Y}^2} \cdot c \\
Z_{(j)}^\intercal Y = \widehat{\beta}_{i,Y} \cdot \frac{(n-1)\cdot {\rm se}^2 (\widehat{\beta}_{1,Y}) + \widehat{\beta}_{1,Y}^2}{(n-1)\cdot {\rm se}^2 (\widehat{\beta}_{i,Y}) + \widehat{\beta}_{i,Y}^2} \cdot c \\
Z_{(j)}^\intercal D = \widehat{\beta}_{i,D} \cdot \frac{(n-1)\cdot {\rm se}^2 (\widehat{\beta}_{1,Y}) + \widehat{\beta}_{1,Y}^2}{(n-1)\cdot {\rm se}^2 (\widehat{\beta}_{i,Y}) + \widehat{\beta}_{i,Y}^2} \cdot c
\end{align*}
where $c$ is an unknown constant factor, i.e. they are known up to an unknown factor $c$.
\end{lemma}
Lemma \ref{lem:sumstat} shows that under conditions (i) and (ii), we can rewrite the quantities $Z^\intercal Y, Z^\intercal D, Z^\intercal Z$ as functions of the summary statistics, up to an unknown constant factor $c$; later, in Theorem \ref{thm:sum_infer}, we show that this unknown constant factor does not impact our inference procedure of $\beta^*$. The first condition (i) of Lemma \ref{lem:sumstat} requires that instruments are pairwise uncorrelated. This is a common assumption in summarized MR data (cf. Section 2 of \citet{bowden2017framework} or Assumption 1 and Section 2.2 of \citet{zhao2018statistical}) and is typically enforced by keeping the genetic distance between each instruments to be sufficiently large and changing the clumping thresholding in software \citep{purcell_plink_2007,hemani2018mr}. The second condition (ii) of Lemma \ref{lem:sumstat} requires that the individual level data be centered, similar to many works on the Lasso, including the instrument selection procedure analyzed by \citet{windmeijer2016use} (Section 5.1, page 7). We remark that from well-known properties of regression, summary statistics do not change when the data is centered or not, and hence, (ii) is not an issue for summarized data from MR. 
Finally, we remark that unlike typical MR methods which assume that the summary statistics of $Y$ and $D$, specifically $\widehat{\beta}_{j,Y}$ and $\widehat{\beta}_{j,D}$, are independent (see Table 5 of \citet{hartwig2017robust} for examples), we do not make this assumption in our setting; in other words, our samples for the outcome and the exposure can overlap. 

The following theorem states that if the conditions of Lemma \ref{lem:sumstat} holds, then we can use summarized statistics to obtain conditional null densities using the same method outlined above.
\begin{theorem}[Inference with summarized data] \label{thm:sum_infer} Under the assumptions in Lemma \ref{lem:sumstat},  the conditional inference density in \eqref{eq:tsls_cden} is identical, regardless of the value of $c$.
\end{theorem}

\section{Application} \label{sec:application}
We demonstrate our method by studying the effect of adiposity, measured using the body mass index (BMI), on diastolic blood pressure (DBP) using MR. A MR study by \citet{timpson2009does} reaffirmed many non-MR observational studies that there is a positive relationship between BMI and DBP, using two SNPs. Our goal is to replicate this analysis by using the UK Biobank \citep{sudlow_UK_2015} and calibrate p-values and confidence intervals for the exposure effect after selecting valid instruments via sisVIVE.

We consider a pool of $p = 315$ SNPs as instruments for BMI and DBP; these SNPs were originally identified by \citet{locke_genetic_2015} and we extract the same set of SNPs in the UK Biobank with the software MR-Base \citep{hemani2018mr}. We use the standard defaults for MR-Base, which makes sure that the instruments are far apart in the genome and they are sufficiently strong. We use the summary statistics provided by MR-Base, specifically the OLS regression coefficient estimates and their standard errors  between each instrument and the treatment / outcome. Around 5\% of individuals did not have both the treatment and the outcome values recorded in the UK Biobank data, leading to slightly different sample sizes for the instrument-treatment regression ($n_1 = 336,107$) and the instrument-outcome regression ($n_2 = 317,756$). Because the difference in the sample sizes are small compared to the total sample size, we take the average of the two samples $n = \frac{1}{2}(n_1 +n_2)$ to be the effective sample size for our analysis. 

We evaluate our method based on multiple choices of $\lambda$. In particular, unlike typical penalized regression where this is chosen automatically with cross-validation, in the empirical analysis, we treat $\lambda$ as a sensitivity parameter where, roughly, a large $\lambda$ would lead to more instruments being selected as valid. Based on the result by \citet{windmeijer2016use} regarding sisVIVE, we only consider sequence of $\lambda$s with selection consistency, typically when the proportion of selected instruments is less than 50\% of the total number of instruments.

The results of our analysis is summarized in figure \ref{fig:mrdata}. We plot the TSLS estimate and the conditional and naive confidence intervals as error bars of the TSLS estimator for different values of $\lambda$. We also denote the number of instruments for each choice of $\lambda$. In general, the TSLS estimate is stable for different choices of selected instruments where $|E| < p/2$ and agrees with previous analysis of the relationship between BMI and blood pressure by \citet{timpson2009does}. Also, the difference between conditional and naive intervals is small for this dataset. This is most likely due to the fact that $n$ is very large and well-known, strong instruments were used, mitigating some of the selection effects. However, as noted before, the difference between the naive and the conditional confidence interval can be large depending on the underlying data generating process and more importantly, our conditional confidence interval has guaranteed marginal and conditional coverage rate whereas the naive confidence interval does not.


\section{Discussion} \label{sec:discussion}
This work demonstrates a conditional approach to inference of the treatment effect after selecting plausibly valid instruments by a data-driven selection method called sisVIVE. 
We show that our conditional p-values and confidence intervals attain the nominal type I error and coverage rate, while the naive ones ignoring the selection effect would perform worse in certain scenarios. We also demonstrate our method through a real dataset analysis. We believe our method can be useful whenever the analyst worries about possible presence of invalid instruments and selection effects from choosing the ``best'' set of instruments.

Our approach outlined above can be extended to correct for the selection effect outside of sisVIVE. In particular, our method can be applied to other instrument selection process as long as it is expressive as a convex program \citep{bi_inferactive_2017}. Also, based on the selective CLT, other model specifications are possible as long as the test statistic of interest is asymptotically Gaussian. 
Finally, while we focused on applications to MR, in the supplementary materials, we demonstrate an application of our method to development economics.



\begin{figure}[!ht]
\begin{center}
\centerline{\includegraphics[width=1.3\columnwidth]{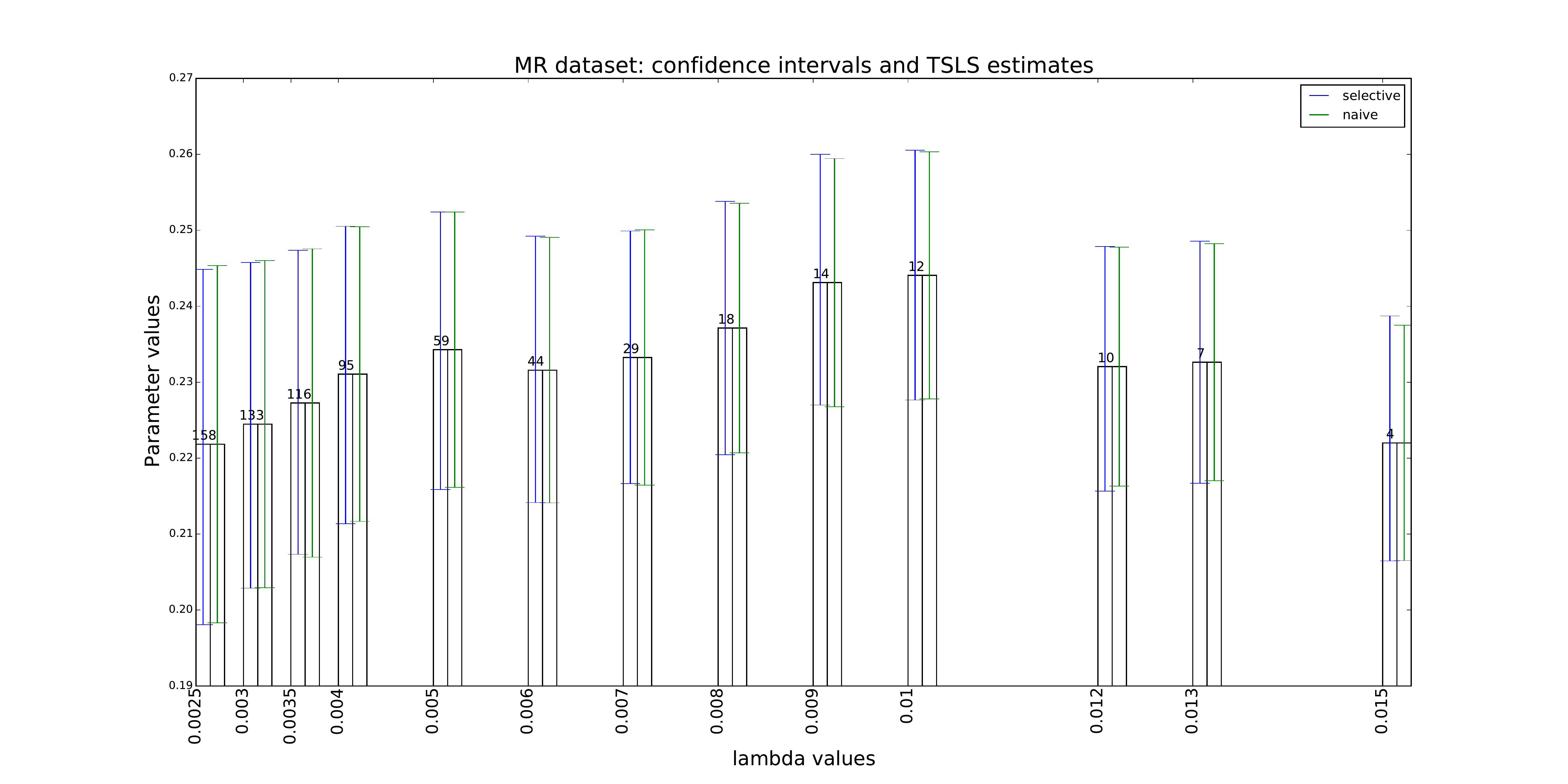}}
\caption{Analysis of the effect of BMI on blood pressure using the UK Biobank. The y-axis is the TSLS estimates and the conditional and naive confidence intervals are plotted as error bars. The x-axis is the value of $\lambda$ and the number of the selected invalid instruments for each $\lambda$ is attached on top of the each bar. }
\label{fig:mrdata}
\end{center}
\end{figure}

\section{Acknowledgements} 
The research of Hyunseung Kang was supported in part by NSF Grant DMS-1811414


\newpage

\begin{center}
{\large\bf SUPPLEMENTARY MATERIALS}
\end{center}

\begin{description}

\item[Proofs:] The file contains all the technical details, including review of model properties in the invalid IV setting, all the proofs to the statements in the article, additional sampling strategy detail; other simulation study and a developmental economics dataset study.

\item[Python package:] See the python software code on github: \url{https://github.com/nanbi/Python-software/tree/lasso_iv}. It is branched from \url{https://github.com/jonathan-taylor/selective-inference}.

\item[UK Biobank dataset:] The dataset used in Section \ref{sec:application} can be freely obtained from here: \url{http://app.mrbase.org/} after agreeing to the data access agreement.

\end{description}

\section{Review of Model Properties}

\subsection{Reduced-Form Model}
Let $E^* = {\rm supp}(\alpha^*)$.
The structural model can be rewritten as a \emph{reduced-form model} where $Y$ and $D$ are only functions of $Z$. This is achieved by substituting the expression for $D$ in the model for $Y$ and $D$
\begin{align*}
Y_i 
&= Z_{i,E}^\intercal (\alpha_{E^*}^* + \gamma_{E^*}^* \beta^*) + Z_{i,-{E^*}}^\intercal \gamma_{-{E^*}} \beta^* + \xi_{i1} \\
D_i &=  Z_{i}^\intercal \gamma^* + \xi_{i2} 
\end{align*}
The error terms for the reduced-form model, $\xi_{i1}$ and $\xi_{i2}$, are defined as $\xi_{i1} = \epsilon_i + \delta_{i2} \beta^*$ and $\xi_{i2} = \delta_i$. The two error terms follow a distribution $(\xi_{i1}, \xi_{i2}) \iid H(0,\Omega^*)$ where $\Omega^*$ is the covariance matrix 
\[
\Omega^* = \begin{pmatrix} 1 & \beta^* \\ 0 & 1 \end{pmatrix} \Sigma^* \begin{pmatrix} 1 & 0 \\ \beta^* & 1 \end{pmatrix}
\]
We note that the transformation between the covariance matrices $\Sigma^*$ and $\Omega^*$ is invertible for any value of $\beta^*$ since we can write $\Omega^*$ as an invertible map of $\Sigma^*$ above.

\subsection{Point Estimation for Model Parameters}
Lemma \ref{lem:tsls} shows that the most popular estimators for $\beta^*$, the two-stage least squares (TSLS), consistently estimates the treatment effect parameter $\beta^*$. We remark that the conditions in Lemma \ref{lem:tsls} are more general than the model conditions in the main manuscript.
\begin{lemma} \label{lem:tsls} Suppose $Z^\intercal \delta /n \inP 0$, $Z^\intercal \xi_{.2}/n \inP 0$, $Z^\intercal Z /n \inP \Lambda$ where $\Lambda$ is a non-singular matrix. For a set $E$ where $E^* \subseteq E^C$ and $\gamma_{E}^* \neq 0$, the stage least squares estimator (TSLS) $\hat{\beta}_{TSLS}$
\begin{equation} \label{eq:tsls_est}
\hat{\beta}_{TSLS} = \frac{D^\intercal (P_Z - P_{Z_E})Y}{D^\intercal (P_Z - P_{Z_E}) D} 
\end{equation}
 is consistent for $\beta^*$ 
\end{lemma}

\begin{proof}[Proof of Lemma \ref{lem:tsls}]
Pugging in the true model, it simplifies to
\begin{align*}
\hat{\beta}_{TSLS} &= \frac{D^\intercal (P_Z - P_{Z_E}) (D \beta^* + Z_E \alpha^*_E + \delta)}{D^\intercal (P_Z - P_{Z_E}) D} \\
&= \beta^* + \frac{D^\intercal (P_Z - P_{Z_E}) \delta}{D^\intercal (P_Z - P_{Z_E}) D} \\
\end{align*}
Now decompose the projection matrix $P_Z - P_{Z_E}$ as a new projection matrix $P_{W}$ where $W$ is a matrix whose columns represent a $L - |E|$ basis vectors of the column space $Z_{-E}$ adjusted by $P_{Z_E}$.
\begin{align*}
\frac{1}{\sqrt{n}} D^\intercal (P_Z - P_{Z_E})\delta &= \frac{1}{\sqrt{n}} D^\intercal (W (W^\intercal W)^{-1} W^\intercal) \delta \\
&= \frac{1}{n} D^\intercal W \left(\frac{W^\intercal W}{n}\right)^{-1} \frac{1}{\sqrt{n}} W^\intercal \delta \\
&\to E(D_i W_i^\intercal) E(W_i W_i^\intercal)^{-1} \frac{1}{\sqrt{n}} W^\intercal \delta \\
&\to E(D_i W_i^\intercal) E(W_i W_i^\intercal)^{-1} \cdot N(0, E(W_i W_i^\intercal) \Sigma_{11}^{*}) \\
&= N(0, \Sigma_{11}^{*}  E(D_i W_i^\intercal) E(W_i W_i^\intercal)^{-1} E(D_i W_i) )
\end{align*}
\begin{align*}
\frac{1}{n} D^\intercal (P_Z - P_{Z_E}) D &= \frac{1}{n} D^\intercal W \left(\frac{W^\intercal W}{n}\right) \frac{W^\intercal D}{n} \\
&\to E(D_i W_i^\intercal) E(W_i W_i^\intercal) E(D_i W_i^\intercal)
\end{align*}
we finally have
$$
\sqrt{n} (\hat{\beta}_{TSLS} - \beta^*) \rightarrow N(0, \frac{\Sigma^*_{11}}{E(D_i W_i^\intercal) E(W_i W_i^\intercal) E(D_i W_i^\intercal)})
$$
which gives the consistency result.
\end{proof}

\begin{lemma} 
\label{lem:sigma}
Suppose $Z^\intercal \delta /n \inP 0$, $Z^\intercal \xi_{.2}/n \inP 0$ and $Z^\intercal Z /n \inP \Lambda$ where $\Lambda$ is a non-singular matrix. Then, under the null hypothesis $H_0: \beta^* = \beta_0$, the estimator $\hat{\Sigma}(\beta_0)$ 
\begin{equation} \label{eq:sigma_est}
\hat{\Sigma}(\beta_0) =  \frac{1}{n} \begin{pmatrix} (Y - D\beta_0)^\intercal \\ D^\intercal \end{pmatrix}  P_{Z^\perp} \begin{pmatrix} Y - D\beta_0 & D \end{pmatrix} 
\end{equation}
is consistent for $\Sigma^*$. Also, the estimator $\hat{\Omega}(\beta_0)$
\begin{equation} \label{eq:omega_est}
\hat{\Omega}(\beta_0) = \begin{pmatrix} 1 & \beta_0 \\ 0 & 1 \end{pmatrix} \hat{\Sigma} \begin{pmatrix} 1 & 0 \\ \beta_0 & 1 \end{pmatrix}
\end{equation}
is consistent for $\Omega^*$.
\end{lemma}

\begin{proof}[Proof of Lemma \ref{lem:sigma}]
Under the null hypothesis $H_0: \beta^* = \beta_0$, we have
\begin{align*}
\hat{\Sigma}_{11}(\beta_0) &= \frac{\| P_{Z^\perp} (Y - D \beta_0) \|_2^2}{n} = \frac{\| P_{Z^\perp} (Z\alpha^* + \delta) \|_2^2}{n} =  \frac{\delta^\intercal \delta}{n} - \frac{\delta^\intercal Z}{n} \left(\frac{Z^\intercal Z}{n}\right)^{-1} \frac{Z^\intercal \delta}{n} \\
\hat{\Sigma}_{12}(\beta_0) &= \frac{(Z\gamma^* + \xi_{.2})^\intercal P_{Z^\perp}(Z\alpha^* + \delta)}{n} = \frac{\xi_{.2}^\intercal P_{Z^\perp} \delta}{n} = \frac{\xi_{.2}^\intercal \delta}{n} - \frac{\xi_{.2}^\intercal Z}{n} \left(\frac{Z^\intercal Z}{n}  \right)^{-1} \frac{Z^\intercal \delta}{n}\\
\hat{\Sigma}_{22} (\beta_0) &= \frac{(Z\alpha^* + \xi_{.2})^\intercal P_{Z^\perp} (Z\alpha^* + \xi_{.2})}{n} = \frac{\xi_{.2}^\intercal P_{Z^\perp} \xi_{.2}}{n} = \frac{\xi_{.2}^\intercal \xi_{.2}}{n} - \frac{\xi_{.2}^\intercal Z}{n} \left(\frac{Z^\intercal Z}{n}  \right)^{-1} \frac{Z^\intercal \xi_{.2}}{n}
\end{align*}
By the law of large numbers, $\delta^\intercal \delta/n \inP \Sigma_{11}^*$, $\xi_{.2}^\intercal \delta/n \inP \Sigma_{12}^*$, and $\xi_{.2}^\intercal \xi_{.2} /n \inP \Sigma_{22}^*$. Also, by the assumed limiting properties of $Z^\intercal Z$, $Z^\intercal \delta$, and $Z^\intercal \xi_{.2}$, we arrive at the consistent estimator of $\Sigma^*$. 

To prove $\hat{\Omega}(\beta_0)$ is a consistent estimator of $\Omega^*$, we simply apply Slutsky's theorem to $\hat{\Sigma}(\beta_0)$ to obtain the desired result.
\end{proof}

\begin{lemma}  \label{lem:consistent_cov}
Under the same conditions as lemma \ref{lem:sigma}, the estimator $\hat{\Sigma}(\hat{\beta}_{TSLS})$
$$
\hat{\Sigma}(\hat{\beta}_{TSLS}) =  \frac{1}{n} \begin{pmatrix} (Y - D\hat{\beta}_{TSLS})^\intercal \\ D^\intercal \end{pmatrix}  P_{Z^\perp} \begin{pmatrix} Y - D\hat{\beta}_{TSLS} & D \end{pmatrix} 
$$
 is consistent for $\Sigma^*$. This also implies that under the null, the estimator $\hat{\Omega}(\hat{\beta}_{TSLS})$ 
 \[
\hat{\Omega}(\hat{\beta}_{TSLS}) = \begin{pmatrix} 1 & \hat{\beta}_{TSLS} \\ 0 & 1 \end{pmatrix} \hat{\Sigma} \begin{pmatrix} 1 & 0 \\ \hat{\beta}_{TSLS} & 1 \end{pmatrix}
\]
 is consistent for $\Omega^*$.
\end{lemma}
\begin{proof}[Proof of Lemma \ref{lem:consistent_cov}]

To prove $\hat{\Sigma}(\hat{\beta}_{TSLS})$ is a consistent estimator of $\Sigma^*$, we apply Slutsky's theorem to $\hat{\Sigma}(\beta_0)$ together with the fact that $\hat{\beta}_{TSLS}$ is consistent for $\beta_0$ to obtain the desired result.

Again by Slutsky's theorem on the consistency of $\hat{\Sigma}(\hat{\beta}_{TSLS})$ we get the consistency of $\hat{\Omega}(\hat{\beta}_{TSLS})$ for $\Omega^*$.
\end{proof}


\section{Proofs of the Conditional Densities} \label{sec:proof}

\subsection{Exact Conditional Density}

\begin{theorem}[Exact conditional density via reparametrization] \label{thm:exact} 
The conditional density of $Y, D, \omega$ can be expressed (up to a proportionality constant) with respect to the variables $Y, D, \beta,\alpha_E,u_{-E}$, i.e. 
\begin{equation} \label{eq:density}
\begin{split}
&\ell (Y,D,\omega \mid {\rm supp}(\alpha) = E, {\rm sign}(\alpha_{E}) = \widehat{s}_E )  \\
\propto&
f (Y, D)\cdot g\left(- \left(\begin{array}{l}
D^\intercal \\
Z^\intercal \\
\end{array}\right) P_Z (Y - Z \alpha - D\beta)  + \lambda 
\left(\begin{array}{l}
u\\
0\\
\end{array}\right) + \epsilon \left(\begin{array}{l}
\alpha\\
\beta\\
\end{array}
\right)  \right) \cdot |\mathcal{J}| \cdot \mathbb{I}(\mathcal{B})
\end{split}
\end{equation}
where 
\begin{align*}
|\mathcal{J}| &= \det \begin{pmatrix} \begin{pmatrix}D^\intercal P_Z D  & D^\intercal Z_E  \\
Z_E^\intercal D & Z_E^\intercal Z_E \end{pmatrix} + \epsilon I \end{pmatrix} \cdot \lambda^{|-E|} \\
\mathcal{B} &= \{ {\rm sign}(\alpha_{E}) = \widehat{s}_{E}, \alpha_{-E} = 0, u_E = \widehat{s}_E, \|u_{-E}\|_\infty \leq 1 \}
\end{align*}
\end{theorem}

\begin{proof}[Proof of Theorem \ref{thm:exact}]
First, we can construct the reparametrization map from KKT condition
\begin{align*}
\omega &= 
- \left(\begin{array}{l}
D^\intercal \\
Z^\intercal \\
\end{array}\right) P_Z (Y - Z \alpha - D\beta)  + \lambda 
\left(\begin{array}{l}
u\\
0\\
\end{array}\right) + \epsilon \left(\begin{array}{l}
\alpha\\
\beta\\
\end{array}
\right) 
\end{align*}
such that
$$
\rm sign (\alpha_E) = u_E = \widehat{s}_E, \quad{} \alpha_{-E} = 0, \quad{} ||u||_{\infty} \leq 1
$$

Second, we compute the Jacobian of this mapping $(Y,D, \omega) \rightarrow (Y,D,\beta, \alpha_E, u_{-E})$:
\begin{align*}
|\mathcal{J}| &= \left|\frac{\partial (Y, D, \omega)}{\partial (Y, D, \beta,\alpha_E,u_{-E})}\right| = \left|\frac{\partial \omega}{\partial (\beta,\alpha_E,u_{-E})}\right| \\
&= 
\det\left(\begin{array}{lll}
D^\intercal P_Z D + \epsilon & D^\intercal Z_E & 0 \\
Z_E^\intercal D & Z_E^\intercal Z_E + \epsilon I_{|E| \times |E|}  & 0 \\
Z_{-E}^\intercal D & Z_{-E}^\intercal Z_E & \lambda I_{|-E| \times |-E|}
\end{array}  \right) \\
&= \det\left( \left(\begin{array}{ll}
D^\intercal P_Z D  & D^\intercal Z_E  \\
Z_E^\intercal D & Z_E^\intercal Z_E 
\end{array} \right) + \epsilon I_{(1 + |E|) \times (1+ |E|)} \right) \cdot \lambda^{|-E|}
\end{align*}

Then, using the change of variables formula, the density of $(Y, D,\omega)$ can be equivalently expressed with variables $(Y, D, \beta,\alpha_E,u_{-E})$ (where $f(Y, D)$ denotes the pre-selection distribution of data)
\begin{align*}
&\ell(Y,D,\omega \mid {\rm supp}(\alpha) = E, \text{sign}(\alpha_{E}) = \widehat{s}_E )  \\
=& f(Y,D) \cdot g(\omega) \cdot {\mathbb{I} ( {\rm supp}(\alpha) = E, \text{sign}(\alpha_{E}) = \widehat{s}_E )} \\
=& f (Y, D)\cdot g\left(- \left(\begin{array}{l}
D^\intercal \\
Z^\intercal \\
\end{array}\right) P_Z (Y - Z \alpha - D\beta)  + \lambda 
\left(\begin{array}{l}
u\\
0\\
\end{array}\right) + \epsilon \left(\begin{array}{l}
\alpha\\
\beta\\
\end{array}
\right)  \right) \cdot |\mathcal{J}| \cdot \mathbb{I}(\mathcal{B})
\end{align*}
where for the last line, the conditioning event  $\{ {\rm supp}(\alpha) = E, \text{sign}(\alpha_{E}) = \widehat{s}_E  \}$ is equivalent to the event $\mathcal{B}$ under the new parametrization by the optimization variables.
\end{proof}

\subsection{Asymptotic Conditional Densities with Selective CLT}  \label{sec:asymptotic}

For brevity we directly describe the Selective CLT in our model and the sisVIVE procedure.

Consider any test statistic $T$ for $\beta^*$ and let $S$ be defined as part of the score of the loss function in sisVIVE:
$$
S = \begin{pmatrix}  Z^\intercal Y \\ D^\intercal P_Z Y  \end{pmatrix}
$$ 
Suppose for any fixed $E$ and under the null $H_0: \beta^* = \beta_0$, $T$ and $\frac{1}{\sqrt{n}} S$ jointly converge to a Normal distribution
\begin{equation} \label{eq:pre_CLT}
\begin{pmatrix}T \\ \frac{1}{\sqrt{n}} S\end{pmatrix} \rightarrow N \left(\begin{pmatrix}\mu_T \\ \mu_S\end{pmatrix}, \begin{pmatrix}\Sigma_T & \Sigma_{T, S} \\ \Sigma_{T, S} & \Sigma_S\end{pmatrix}\right), \quad{} n \rightarrow \infty
\end{equation}
If $\Sigma_T$ and $\Sigma_{T,S}$ are known, by the joint gaussianity of $T$ and $S$, we can express $S$ into mutually orthogonal parts, $F$ and $\Sigma_{S, T} \Sigma^{-1}_{T} T$ where these two terms are independent of each other 
\begin{align*}
S &= F + \Sigma_{S, T} \Sigma^{-1}_{T} T \\
F &= \begin{pmatrix} Z^\intercal Y \\ D^\intercal P_Z Y \end{pmatrix} - \Sigma_{S,T}\Sigma_T^{-1} T
\end{align*}
Also, by the joint gaussianity, $F$ is, in the asymptotic sense, a sufficient statistic for $\mu_S$ and by standard arguments for exponential families, conditioning on $F$ in the conditional density \eqref{eq:density} will lead to a free-of-$\mu_S$ density, i.e.
\begin{equation} \label{eq:asymp_density}
\ell_{H_0} (T, \omega \mid {\rm supp}(\alpha) = E, \text{sign}(\alpha_{E}) = \widehat{s}_E, F) = \phi_{(\mu_T, \Sigma_T )} (T) \cdot g(T, \beta, \alpha, u \mid E, \widehat{s}_E, F) \cdot |\mathcal{J}| \cdot \mathbb{I}(\mathcal{B}) 
\end{equation}
where $\phi(\mu_T,\Sigma_T)$ is the density of the Normal distribution with mean $\mu_T$ and variance $\Sigma_{T}$. Intuitively, the terms on the right hand side of $\phi_{(\mu_T, \Sigma_T)}(T)$ in \eqref{eq:asymp_density} reflect that the effect that instrument selection has on the unconditional density of the test statistic $\phi_{(\mu_T, \Sigma_T )} (T)$, similar to the terms on the right hand size of $f(Y,D)$ in \eqref{eq:density}. 
But,  \eqref{eq:asymp_density} is easier to sample from than \eqref{eq:density} because by conditioning on $F$, the \eqref{eq:asymp_density} is only function of parameters $\mu_T, \Sigma_T$, which are known from the null distribution of the test statistic $T$. Also, if $\mathcal{L}(T\mid E, \widehat{s}_E, F)$ is the distribution of $T$ from \eqref{eq:asymp_density} by marginalizing out $\omega$, then \citet{markovic_bootstrap_2017} has proved under regularity conditions, the p-value based on $\mathcal{L}(T\mid E, \widehat{s}_E, F)$ is uniformly distributed under $H_0$ and thus, with the conditioning of $F$, we obtain a pivotal density to sample from. Importantly, the result still holds under regularity conditions if we replace the unknown $\Sigma_{T}$ and $\Sigma_{T,S}$ with consistent estimators of them \citep{tian_randomized_2016}; see the next section for specific instances of $T$ being the TSLS and the AR test statistics.


\subsection{Asymptotic Conditional Densities with TSLS Statistics}

To discriminate from the random variables, we explicitly use the supscription {\emph{$\rm obs$}} to denote the observed values plugged in from the dataset. 

There are several preparation steps regarding the components in the conditional density and theorem \ref{thm:tsls_asymp} combines them to finally arrive at the conditional density for sampling. Lemma \ref{lem:tsls_asymp} points out the components in $T_{TSLS}$ that determines its asymptotic behavior. Lemma \ref{lem:tsls_kkt_asymp} deals with the asymptotics of the score vector in the KKT condition.
\begin{lemma} \label{lem:tsls_asymp}
Assuming $\hat{\Sigma}_{11}$ is consistent estimator that $\hat{\Sigma}_{11} = \Sigma^*_{11} + O_p (1)$, then
\begin{equation}
T_{TSLS} = \left[\frac{D^\intercal Z (Z^\intercal Z)^{-1}}{\sqrt{\hat{\Sigma}_{11}} \sqrt{D^\intercal (P_{Z} - P_{Z_{E}}) D}}\right]^{obs}\cdot (Z^\intercal \delta) - \left[\frac{D^\intercal Z_{E} (Z_{E}^\intercal Z_{E})^{-1}}{\sqrt{\hat{\Sigma}_{11}} \sqrt{D^\intercal (P_{Z} - P_{Z_{E}}) D}} \right]^{obs}\cdot(Z_E^\intercal \delta) + O_p (1)
\end{equation}
\end{lemma}

\begin{lemma}\label{lem:tsls_kkt_asymp}
The reconstruction map is asymptotically equivalent to
$$
\omega = \begin{pmatrix} Z^\intercal Z & Z^\intercal D \\ (Z^\intercal D)^\intercal & D^\intercal P_Z D \end{pmatrix}^{obs} \begin{pmatrix} \alpha - \alpha^* \\ \beta - \beta^* \end{pmatrix} - \begin{pmatrix} I \\ (Z^\intercal D)^\intercal (Z^\intercal Z)^{-1} \end{pmatrix}^{obs} \cdot Z^\intercal \delta + \lambda \begin{pmatrix} u \\ 0 \end{pmatrix} + \epsilon \begin{pmatrix} \alpha \\ \beta \end{pmatrix} + O_p (1)
$$
\end{lemma}

We can see that the term $Z^\intercal \delta$ dictates the asymptotic behavior of both the test statistic and the KKT condition, which essentially allows us to have a linear decomposition with explicit parametric expression. Theorem \ref{thm:tsls_asymp} gives the final conditional density using the test statistic $T_{TSLS}$ that is feasible to sampling.

\begin{theorem}[Asymptotic conditional density of TSLS test statistic] \label{thm:tsls_asymp}
The asymptotic conditional density of $T_{\rm TSLS}$ under the null hypothesis $H_0: \beta^* = \beta_0$ can be expressed (up to a proportionality constant) with respect to the variables $T_{\rm TSLS}, \beta, \alpha_E, u_{-E}$:
\begin{equation} \label{eq:tsls_cden}
\begin{split} 
&\ell_{\beta_0}(T_{\rm TSLS}, \omega \mid {\rm supp}(\alpha) = E, \text{sign}(\alpha_{E}) = \widehat{s}_E, F) \propto \phi_{(0, 1)} (T_{\rm TSLS}) \\
&\cdot g\left\{- \Sigma_{S,T}^{obs}\cdot T_{\rm TSLS}  + \left[\begin{pmatrix}Z^\intercal Z & Z^\intercal D \\ D^\intercal Z & D^\intercal P_{Z} D\end{pmatrix} + \epsilon \cdot I \right]^{obs}  \begin{pmatrix}\alpha \\ \beta \end{pmatrix} + \lambda \begin{pmatrix} u \\ 0 \end{pmatrix} - \left[\begin{pmatrix}Z^\intercal Y \\ D^\intercal P_Z Y\end{pmatrix} - \Sigma_{S,T}\cdot \widehat{T}_{TSLS}\right]^{obs} \right\} \\
&  \cdot |\mathcal{J}|^{obs} \cdot \mathbb{I}(\mathcal{B}) 
\end{split}
\end{equation}
where 
$$
\Sigma_{S,T} =  \sqrt{\frac{\hat{\Sigma}_{11}}{D^\intercal (P_Z - P_{Z_E}) D}} \begin{pmatrix} Z^\intercal (P_Z - P_{Z_E}) D \\ D^\intercal (P_Z - P_{Z_E}) D\end{pmatrix}
$$
\end{theorem}

We remark that $\Sigma^*_{12}$, i.e. the correlation between $D$ and $\delta$ does not show up in the conditional density.

We can also use the TSLS estimator in the conditional density instead of the TSLS test statistic. They are essentially equivalent asymptotically in the inference. The test statistic $T_{TSLS}(\beta_0)$ and the point estimator $\beta_{TSLS}$ are related as follows: 
\[
T_{TSLS}(\beta_0) = \frac{\beta_{TSLS} - \beta_0}{\sqrt{\frac{\Sigma_{11}^*}{D^\intercal (P_Z - P_{Z_E}) D}}} 
\]

\begin{corollary}[Asymptotic conditional density of TSLS estimator] \label{thm:tsls_beta_asymp}
The asymptotic conditional density of $\beta_{\rm TSLS}$ under the null hypothesis $H_0: \beta^* = \beta_0$ can be expressed (up to a proportionality constant) with respect to the variables $\beta_{\rm TSLS}, \beta, \alpha_E, u_{-E}$:
\begin{equation}
\begin{split}
&\ell_{\beta_0} (\beta_{\rm TSLS}, \omega \mid {\rm supp}(\alpha) = E, \text{sign}(\alpha_{E}) = \widehat{s}_E, F) \propto \phi_{\left(\beta_0, \Sigma_T^{obs}\right)} (\beta_{\rm TSLS}) \\
& \cdot g\left\{-\Sigma_{S,T}^{obs} \cdot \Sigma_T^{-1, obs} \cdot \beta_{\rm TSLS} + \left[\begin{pmatrix}Z^\intercal Z & Z^\intercal D \\ D^\intercal Z & D^\intercal P_Z D\end{pmatrix} + \epsilon\cdot I\right]^{obs}  \begin{pmatrix}\alpha \\ \beta\end{pmatrix} + \lambda \begin{pmatrix} u \\ 0 \end{pmatrix} - \left[\begin{pmatrix}Z^\intercal Y \\ D^\intercal P_Z Y\end{pmatrix}-\Sigma_{S,T}\Sigma_T^{-1, obs}\widehat{\beta}_{\rm TSLS}\right]^{obs} \right\} \\
& \cdot |\mathcal{J}|^{obs} \cdot \mathbb{I}(\mathcal{B}) 
\end{split}
\end{equation}
where
$$
\Sigma_T = \frac{\hat{\Sigma}_{11}}{D^\intercal (P_Z - P_{Z_E}) D},\quad{}\Sigma_{S,T} = \frac{\hat{\Sigma}_{11}}{D^\intercal (P_Z - P_{Z_E}) D}\begin{pmatrix}Z^\intercal (P_Z - P_{Z_E}) D \\ D^\intercal (P_Z - P_{Z_E}) D\end{pmatrix}
$$
\end{corollary}

\begin{proof}[Proof of Lemma \ref{lem:tsls_asymp}]
Notice we have
\begin{align*}
T_{TSLS} (\beta_0) &= \frac{D^\intercal (P_Z - P_{Z_E}) (Y - D\beta_0)}{\sqrt{\hat{\Sigma}_{11}} \sqrt{D^\intercal (P_Z - P_{Z_E}) D}} \\
&= \frac{D^\intercal (P_Z - P_{Z_E}) \delta}{\sqrt{\hat{\Sigma}_{11}} \cdot \sqrt{D^\intercal (P_Z - P_{Z_E}) D}} \\
&= \frac{\frac{1}{\sqrt{n}}D^\intercal (P_Z - P_{Z_E}) \delta}{\sqrt{\hat{\Sigma}_{11}} \cdot \sqrt{\frac{1}{n} D^\intercal (P_Z - P_{Z_E}) D}}
\end{align*}
where the second line is by plugging in the true model.

Let us denote for simplification the quadratic terms
$$
\hat{A} = \frac{Z^\intercal Z}{n} , \quad \hat{A}_{E} = \frac{Z_E^\intercal Z_E}{n}, \quad \hat{B} = \frac{Z^\intercal D}{n}, \quad \hat{B}_{E} = \frac{Z_E^\intercal D}{n}
$$
and their asymptotic means all exist
$$
A = E(Z_i \cdot Z_i^\intercal), \quad A_E = E(Z_{i,E} \cdot Z_{i,E}^\intercal), \quad B = E(Z_i \cdot D), \quad B_E = E(Z_{i, E} \cdot D)
$$
we know that 
$$
\hat{A} = A + O_p (n^{-\frac{1}{2}}), \quad \hat{A}_E = A_E + O_p (n^{-\frac{1}{2}}), \quad \hat{B} = B + O_p (n^{-\frac{1}{2}}),\quad \hat{B}_E = B_E + O_p (n^{-\frac{1}{2}})
$$
Now the numerator
\begin{align*}
&\quad \frac{1}{\sqrt{n}}D^\intercal (P_Z - P_{Z_E}) \delta \\
&= \frac{1}{\sqrt{n}} (Z^\intercal D)^\intercal (Z^\intercal Z)^{-1} (Z^\intercal \delta) - \frac{1}{\sqrt{n}} (Z^\intercal_E D)^\intercal(Z_E^\intercal Z_E)^{-1} (Z_E \delta) \\
&= (Z^\intercal D)^\intercal(Z^\intercal Z)^{-1} \cdot \frac{1}{\sqrt{n}}(Z^\intercal \delta) - (Z_E^\intercal D)^\intercal (Z_E^\intercal Z_E)^{-1} \cdot \frac{1}{\sqrt{n}} (Z^\intercal_E \delta) \\
&= B^\intercal A^{-1} \cdot \frac{1}{\sqrt{n}}(Z^\intercal \delta) - B^\intercal A_E^{-1} \cdot \frac{1}{\sqrt{n}} (Z^\intercal_E \delta) + (\hat{B}^\intercal \hat{A}^{-1} - B^\intercal A^{-1})\cdot \frac{1}{\sqrt{n}}(Z^\intercal \delta) + (\hat{B}^\intercal \hat{A}_E^{-1} - B^\intercal A_E^{-1}) \cdot \frac{1}{\sqrt{n}} (Z^\intercal_E \delta) \\
& = B^\intercal A^{-1} \cdot \frac{1}{\sqrt{n}}(Z^\intercal \delta) - B^\intercal A_E^{-1} \cdot \frac{1}{\sqrt{n}} (Z^\intercal_E \delta) + O_p (n^{-\frac{1}{2}})
\end{align*}
where the last equation has used the fact that
$$
\hat{A}^{-1} = A^{-1} + O_p (n^{-\frac{1}{2}}), \quad \hat{A}^{-1}_E = A^{-1}_E + O_p (n^{-\frac{1}{2}})
$$
and that
$$
\hat{B}^\intercal \hat{A}^{-1}  = B^\intercal A^{-1} + O_p (n^{-\frac{1}{2}}), \quad \hat{B}^\intercal \hat{A}_E^{-1} =  B^\intercal A_E^{-1} + O_p (n^{-\frac{1}{2}})
$$
from delta method, and also 
$$
\frac{1}{\sqrt{n}} (Z^\intercal \delta) = O_p (1)
$$ 
from standard central limit theorem since $E(Z^\intercal \delta) = 0$ under the model assumption.

The denominator serves as the normalizing factor with similar argument
\begin{align*}
\sqrt{\hat{\Sigma}_{11}} \sqrt{\frac{1}{n} D^\intercal (P_Z - P_{Z_E}) D} &= (\sqrt{\Sigma^*_{11}} + o_p (1))\sqrt{\frac{1}{n} (Z^\intercal D)^\intercal (Z^\intercal Z)^{-1} (Z^\intercal D) - \frac{1}{n} (Z_E^\intercal D)^\intercal (Z_E^\intercal Z_E) (Z_E^\intercal D)} \\
&= (\sqrt{\Sigma^*_{11}} + o_p (1))\sqrt{\hat{B}^\intercal \hat{A}^{-1} \hat{B} - \hat{B}_E^\intercal \hat{A}_E^{-1} \hat{B}_E} \\
&= (\sqrt{\Sigma^*_{11}} + o_p (1)) \cdot (\sqrt{B^\intercal A^{-1} B - B_E^\intercal A_E^{-1} B_E} + o_p (1)) \\
&= \sqrt{\Sigma^*_{11}} \sqrt{B^\intercal A^{-1} B - B_E^\intercal A_E^{-1} B_E} + o_p (1)
\end{align*}
Therefore the test statistic is asymptotically
\begin{align*}
T_{TSLS}(\beta_0) &= \frac{B^\intercal A^{-1}}{ \sqrt{\Sigma^*_{11}} \sqrt{B^\intercal A^{-1} B - B_E^\intercal A_E^{-1} B_E}} \frac{1}{\sqrt{n}} \cdot (Z^\intercal \delta) - \frac{B^\intercal A_E^{-1}}{\sqrt{\Sigma^*_{11}} \sqrt{B^\intercal A^{-1} B - B_E^\intercal A_E^{-1} B_E}} \frac{1}{\sqrt{n}} \cdot (Z^\intercal_E \delta) + o_p (1)\\
&= \left[\frac{D^\intercal Z (Z^\intercal Z)^{-1}}{\sqrt{\hat{\Sigma}_{11}} \sqrt{D^\intercal (P_{Z} - P_{Z_{E}}) D}}\right]^{obs}\cdot (Z^\intercal \delta) - \left[\frac{D^\intercal Z_{E} (Z_{E}^\intercal Z_{E})^{-1}}{\sqrt{\hat{\Sigma}_{11}} \sqrt{D^\intercal (P_{Z} - P_{Z_{E}}) D}} \right]^{obs}\cdot(Z_E^\intercal \delta) + o_p (1) \\
\end{align*}
The second line follows from the same argument replacing the sampling variables $\hat{A}$ with the observed values $\left[\frac{Z^\intercal Z}{n}\right]^{obs}$, i.e. $\left[\frac{Z^\intercal Z}{n}\right]^{obs} = A + O_p (n^{-\frac{1}{2}})$, and with $\hat{A}_E, \hat{B}, \hat{B}_E$.
\end{proof}

\begin{proof}[Proof of Lemma \ref{lem:tsls_kkt_asymp}]
We start from the KKT condition
\begin{align*}
\omega &= \begin{pmatrix} Z^\intercal \\ D^\intercal \end{pmatrix} P_Z \begin{pmatrix} Z & D \end{pmatrix} \begin{pmatrix} \alpha \\ \beta\end{pmatrix} - \begin{pmatrix} Z^\intercal \\ D^\intercal \end{pmatrix} P_Z Y + \lambda \begin{pmatrix} u \\ 0 \end{pmatrix} + \epsilon \begin{pmatrix} \alpha \\ \beta \end{pmatrix}\\
&= \begin{pmatrix} Z^\intercal Z & Z^\intercal D \\ (Z^\intercal D)^\intercal & D^\intercal P_Z D \end{pmatrix} \begin{pmatrix} \alpha - \alpha^* \\ \beta - \beta^* \end{pmatrix} - \begin{pmatrix} Z^\intercal \delta \\ D^\intercal P_Z \delta \end{pmatrix} + \lambda \begin{pmatrix} u \\ 0 \end{pmatrix} + \epsilon \begin{pmatrix} \alpha \\ \beta \end{pmatrix} \\ 
&= n \begin{pmatrix}A & B \\ B^\intercal & B^\intercal A^{-1} B \end{pmatrix} \cdot \begin{pmatrix} \alpha - \alpha^* \\ \beta - \beta^* \end{pmatrix} - \begin{pmatrix} I \\ B^\intercal A^{-1} \end{pmatrix} \cdot Z^\intercal \delta +\lambda \begin{pmatrix} u \\ 0 \end{pmatrix} + \epsilon \begin{pmatrix} \alpha \\ \beta \end{pmatrix} \\
&+ \begin{pmatrix} Z^\intercal Z - nA & Z^\intercal D - nB \\ (Z^\intercal D)^\intercal - n B^\intercal& D^\intercal P_Z D - n B^\intercal A^{-1} B \end{pmatrix} \begin{pmatrix} \alpha - \alpha^* \\ \beta - \beta^* \end{pmatrix} - \begin{pmatrix}0 \\ (Z^\intercal D)^\intercal (Z^\intercal Z)^{-1}  - B^\intercal A^{-1}\end{pmatrix} \cdot Z^\intercal \delta \\
&= n \begin{pmatrix}A & B \\ B^\intercal & B^\intercal A^{-1} B \end{pmatrix} \cdot \begin{pmatrix} \alpha - \alpha^* \\ \beta - \beta^* \end{pmatrix} - \begin{pmatrix} I \\ B^\intercal A^{-1} \end{pmatrix} \cdot Z^\intercal \delta +\lambda \begin{pmatrix} u \\ 0 \end{pmatrix} + \epsilon \begin{pmatrix} \alpha \\ \beta \end{pmatrix} + O_p (1) - O_p (1) \\
&= n \begin{pmatrix}A & B \\ B^\intercal & B^\intercal A^{-1} B \end{pmatrix} \cdot \begin{pmatrix} \alpha - \alpha^* \\ \beta - \beta^* \end{pmatrix} - \begin{pmatrix} I \\ B^\intercal A^{-1} \end{pmatrix} \cdot Z^\intercal \delta +\lambda \begin{pmatrix} u \\ 0 \end{pmatrix} + \epsilon \begin{pmatrix} \alpha \\ \beta \end{pmatrix} + O_p (1)
\end{align*}
where the second line is by plugging in the model.

Now consider both sides of the second line. We choose the randomization with scale $\omega \sim N(0, n \cdot I)$, hence $\omega = O_p (n^{\frac{1}{2}})$. We know $\lambda \begin{pmatrix} u \\ 0 \end{pmatrix}$ and $\epsilon \begin{pmatrix}\alpha \\ \beta\end{pmatrix}$ are both of order $O_p (n^{\frac{1}{2}})$ from the way we set $\lambda$ and $\epsilon$ and that $\frac{1}{\sqrt{n}} Z^\intercal \delta = O(1)$. Comparing the density of both sides of the 2nd equation, we know it has to be true that $\sqrt{n}(\alpha - \alpha^*), \sqrt{n}(\beta - \beta^*) = O(1)$. 

Reorganizing terms gives the third line, and according to the asymptotic arguments above we get to the 4th line. 

The same argument replacing the sampling variables with the observed values gives us the desired result.
\end{proof}

\begin{proof}[Proof of Theorem \ref{thm:tsls_asymp}]

From lemma \ref{lem:tsls_kkt_asymp} denote 
\begin{align*}
T_{TSLS} &= \left[\frac{D^\intercal Z (Z^\intercal Z)^{-1}}{\sqrt{\hat{\Sigma}_{11}} \sqrt{D^\intercal (P_{Z} - P_{Z_{E}}) D}}\right]^{obs}\cdot (Z^\intercal \delta) - \left[\frac{D^\intercal Z_{E} (Z_{E}^\intercal Z_{E})^{-1}}{\sqrt{\hat{\Sigma}_{11}} \sqrt{D^\intercal (P_{Z} - P_{Z_{E}}) D}} \right]^{obs}\cdot(Z_E^\intercal \delta) + o_p (1) \\
&\stackrel{d}{=} C_1 \cdot Z^\intercal \delta - C_2 \cdot Z_E^\intercal \delta + o_p (1)
\end{align*}
and we know its pre-selection asymptotic distribution is gaussian with $\mu_T = 0, \Sigma_T = 1$. The data vector is the part $S$ in the KKT condition 
, i.e. 
\begin{align*}
S &= \begin{pmatrix} I \\ (Z^\intercal D)^\intercal (Z^\intercal Z)^{-1} \end{pmatrix}^{obs} \cdot Z^\intercal \delta + O_p (1) \\
&\stackrel{d}{=} C_3 \cdot Z^\intercal \delta + O_p (1)
\end{align*}
or that
$$
\frac{1}{\sqrt{n}} S = C_3 \cdot Z^\intercal \delta + o_p (1) 
$$
therefore we see that $(T_{TSLS}, \frac{1}{\sqrt{n}} S)$ is asymptotically jointly gaussian, and direct computation gives
\begin{align*}
\Sigma_{S,T} &= \text{Cov} (C_3 \cdot Z^\intercal \delta, C_1 \cdot Z^\intercal \delta - C_2 \cdot Z_E^\intercal \delta) \\
&= E \left([C_3 \cdot Z^\intercal \delta] \cdot [C_1 \cdot Z^\intercal \delta - C_2 \cdot Z_E^\intercal \delta]\right) - E(C_3 \cdot Z^\intercal \delta)\cdot E(C_1 \cdot Z^\intercal \delta - C_2 \cdot Z_E^\intercal \delta) \\
&= E\left([C_3 \cdot Z^\intercal \delta] \cdot [C_1 \cdot Z^\intercal \delta - C_2 \cdot Z_E^\intercal \delta] \mid Z\right) - E\left( (C_3 \cdot Z^\intercal \delta) \mid Z\right) \cdot E\left( (C_1 \cdot Z^\intercal \delta - C_2 \cdot Z_E^\intercal \delta) \mid Z\right) \\
&= \Sigma_{11} * C_3 \cdot C_1 \cdot E(Z^\intercal Z) - \Sigma_{11} * C_3 \cdot C_2 \cdot E(Z^\intercal Z_E) - 0 \\
&= \sqrt{\frac{\hat{\Sigma}_{11}}{D^\intercal (P_Z - P_{Z_E}) D}} \begin{pmatrix} Z^\intercal (P_Z - P_{Z_E}) D \\ D^\intercal (P_Z - P_{Z_E}) D\end{pmatrix}
\end{align*}
where for the second to last line with large sample size we can plug in $Z^\intercal Z$ and $Z^\intercal Z_E$ for the expectations and get rid of the difference because
\begin{align*}
Z^\intercal Z &= E(Z^\intercal Z) + O_p (\sqrt{n}) \\
Z^\intercal Z_E &= E(Z^\intercal Z_E) + O_p (\sqrt{n}) 
\end{align*}
and finally get to the last line.


Now start from lemma \ref{lem:tsls_kkt_asymp} and apply the linear decomposition technique with our $(T, S)$ to derive the asymptotic version of KKT condition
\begin{align*}
\omega &= \begin{pmatrix} Z^\intercal Z & Z^\intercal D \\ (Z^\intercal D)^\intercal & D^\intercal P_Z D \end{pmatrix}^{obs} \begin{pmatrix} \alpha - \alpha^* \\ \beta - \beta^* \end{pmatrix} - \begin{pmatrix} I \\ (Z^\intercal D)^\intercal (Z^\intercal Z)^{-1} \end{pmatrix}^{obs} \cdot Z^\intercal \delta + \lambda \begin{pmatrix} u \\ 0 \end{pmatrix} + \epsilon \begin{pmatrix} \alpha \\ \beta \end{pmatrix} \\
&= - \Sigma_{S,T}^{obs}\cdot T + \begin{pmatrix} Z^\intercal Z & Z^\intercal D \\ (Z^\intercal D)^\intercal & D^\intercal P_Z D \end{pmatrix}^{obs} \begin{pmatrix} \alpha - \alpha^* \\ \beta - \beta^* \end{pmatrix} + \lambda \begin{pmatrix} u \\ 0 \end{pmatrix} + \epsilon \begin{pmatrix} \alpha \\ \beta \end{pmatrix} - \left[\begin{pmatrix} I \\ D^\intercal Z (Z^\intercal Z)^{-1} \end{pmatrix} \cdot Z^\intercal \delta - \Sigma_{S,T}\cdot T\right]^{obs} \\
&= - \Sigma_{S,T}^{obs}\cdot T + \left[\begin{pmatrix}Z^\intercal Z & Z^\intercal D \\ D^\intercal Z & D^\intercal P_{Z} D\end{pmatrix}^{obs} + \epsilon \cdot I \right] \cdot \begin{pmatrix}\alpha \\ \beta \end{pmatrix} + \lambda \begin{pmatrix} u \\ 0 \end{pmatrix} - \left[\begin{pmatrix} I \\ D^\intercal Z (Z^\intercal Z)^{-1} \end{pmatrix} \cdot Z^\intercal Y - \Sigma_{S,T}\cdot T\right]^{obs}
\end{align*}
where in the last line to re-combine the terms with $\alpha^*$ and $\beta^*$ with the observed values $\delta$ which gives back the observed $Y$. This gives the reparametrization mapping $(T, \omega) \rightarrow (T, \alpha_E, \beta, u_{-E})$. Since it is linear and the coefficient matrices are fixed values and the Jacobian is constant. 

Combining the above steps and from standard change-of-variable formula we obtain the conditional sampling density of $(T, \alpha, \beta, u_{-E})$ as
\begin{align*}
&\ell_{\beta_0}(T, \alpha, \beta, u_{-E} \mid {\rm supp}(\hat{\alpha}) = E, \text{sign}(\hat{\alpha}_{E}) = s_E) \\
&\propto \phi_{0, 1} (T)\cdot g\left\{- \Sigma_{S,T}^{obs}\cdot T + \left[\begin{pmatrix}Z^\intercal Z & Z^\intercal D \\ D^\intercal Z & D^\intercal P_{Z} D\end{pmatrix}^{obs} + \epsilon \cdot I \right] \cdot \begin{pmatrix}\alpha \\ \beta \end{pmatrix} + \lambda \begin{pmatrix} u \\ 0 \end{pmatrix} - \left[\begin{pmatrix} I \\ D^\intercal Z (Z^\intercal Z)^{-1} \end{pmatrix} \cdot Z^\intercal Y - \Sigma_{S,T}\cdot T\right]^{obs} \right\} \\
\end{align*}

\end{proof}

\begin{proof}[Proof of Corollary \ref{thm:tsls_beta_asymp}]
The logic of the proof is similar to theorem \ref{thm:tsls_asymp}.

The TSLS estimator 
\begin{align*}
\beta_{\rm TSLS} &= \frac{D^\intercal (P_Z - P_{Z_E})Y}{D^\intercal (P_Z - P_{Z_E}) D} \\
&= \beta_0 + \frac{D^\intercal (P_Z - P_{Z_E})\delta}{D^\intercal (P_Z - P_{Z_E}) D} \\
&= \beta_0 + \left(\frac{D^\intercal Z (Z^\intercal Z)^{-1}}{D^\intercal (P_Z - P_{Z_E})D}\right) \cdot (Z^\intercal \delta) - \left(\frac{D^\intercal Z_E (Z_E^\intercal Z_E)^{-1}}{D^\intercal (P_Z - P_{Z_E}) D}\right) \cdot (Z_E^\intercal \delta) \\
&= \beta_0 + \left(\frac{\frac{D^\intercal Z}{\sqrt{n}} (\frac{Z^\intercal Z}{\sqrt{n}})^{-1}}{\frac{D^\intercal Z}{\sqrt{n}} (\frac{Z^\intercal Z}{\sqrt{n}})^{-1} \frac{Z^\intercal D}{\sqrt{n}}}\right) \cdot (\frac{Z^\intercal \delta}{\sqrt{n}}) - \left(\frac{\frac{D^\intercal Z_E}{\sqrt{n}} (\frac{Z_E^\intercal Z_E}{\sqrt{n}})^{-1}}{\frac{D^\intercal Z}{\sqrt{n}} (\frac{Z^\intercal Z}{\sqrt{n}})^{-1} \frac{Z^\intercal D}{\sqrt{n}}}\right)\cdot (\frac{Z_E^\intercal \delta}{\sqrt{n}}) \\
&= \beta_0 + \left\{\left(\frac{\frac{D^\intercal Z}{\sqrt{n}} (\frac{Z^\intercal Z}{\sqrt{n}})^{-1}}{\frac{D^\intercal Z}{\sqrt{n}} (\frac{Z^\intercal Z}{\sqrt{n}})^{-1} \frac{Z^\intercal D}{\sqrt{n}}}\right)^{obs} + O_p (\frac{1}{\sqrt{n}}) \right\}\cdot (\frac{Z^\intercal \delta}{\sqrt{n}}) - \left\{\left(\frac{\frac{D^\intercal Z_E}{\sqrt{n}} (\frac{Z_E^\intercal Z_E}{\sqrt{n}})^{-1}}{\frac{D^\intercal Z}{\sqrt{n}} (\frac{Z^\intercal Z}{\sqrt{n}})^{-1} \frac{Z^\intercal D}{\sqrt{n}}}\right)^{obs} + O_p (\frac{1}{\sqrt{n}}) \right\}\cdot (\frac{Z_E^\intercal \delta}{\sqrt{n}}) \\
&\approx \beta_0 + \left(\frac{D^\intercal Z (Z^\intercal Z)^{-1}}{D^\intercal (P_Z - P_{Z_E})D}\right)^{obs} \cdot (Z^\intercal \delta) - \left(\frac{D^\intercal Z_E (Z_E^\intercal Z_E)^{-1}}{D^\intercal (P_Z - P_{Z_E}) D}\right)^{obs}\cdot (Z_E^\intercal \delta) \\
&\stackrel{d}{=} \beta_0 + C_1 \cdot (Z^\intercal \delta) - C_2 \cdot (Z_E^\intercal \delta)
\end{align*}
where the asymptotic argument to plug in the observed values is the same as the proof for lemma \ref{thm:tsls_asymp} without $\hat{\Sigma}_{11}$ in the denominators of the terms.

From this last line and the fact that 
\begin{align*}
\frac{Z^\intercal \delta}{\sqrt{n}} &\rightarrow N(0, \Sigma^*_{11}\cdot E(Z_i^\intercal Z_i)) \\
\frac{Z^\intercal_E \delta}{\sqrt{n}} &\rightarrow N(0, \Sigma^*_{11}\cdot E(Z^\intercal_{i,E} Z_{i,E})) \\
Cov(\frac{Z^\intercal \delta}{\sqrt{n}}, \frac{Z^\intercal_E \delta}{\sqrt{n}}) &= \Sigma^*_{11} \cdot E(Z_i^\intercal Z_{i,E}) 
\end{align*}
we can derive that the pre-selection asymptotic distribution of $\beta_{\rm TSLS}$ is gaussian with mean 
$$
\mu_T = \beta_0
$$ 
and approximate variance 
$$
\Sigma_T = \frac{\Sigma^*_{11}}{D^\intercal (P_Z - P_{Z_E}) D}
$$

With the same KKT condition as in lemma \ref{lem:tsls_kkt_asymp}
$$
\omega = \begin{pmatrix} Z^\intercal Z & Z^\intercal D \\ (Z^\intercal D)^\intercal & D^\intercal P_Z D \end{pmatrix}^{obs} \begin{pmatrix} \alpha - \alpha^* \\ \beta - \beta^* \end{pmatrix} - \begin{pmatrix} I \\ (Z^\intercal D)^\intercal (Z^\intercal Z)^{-1} \end{pmatrix}^{obs} \cdot Z^\intercal \delta + \lambda \begin{pmatrix} u \\ 0 \end{pmatrix} + \epsilon \begin{pmatrix} \alpha \\ \beta \end{pmatrix}
$$
and the same data vector 
\begin{align*}
S &= \begin{pmatrix} I \\ (Z^\intercal D)^\intercal (Z^\intercal Z)^{-1} \end{pmatrix}^{obs} \cdot Z^\intercal \delta \\
&\stackrel{d}{=} C_3 \cdot Z^\intercal \delta
\end{align*}

and similar as in the proof of theorem \ref{thm:tsls_asymp} it can be seen that $(\beta_{TSLS}, \frac{1}{\sqrt{n}} S)$ is asymptotically jointly gaussian and to derive
\begin{align*}
\Sigma_{S,T} &= \text{Cov} (C_3 \cdot Z^\intercal \delta, \beta_0 + C_1 \cdot Z^\intercal \delta - C_2 \cdot Z_E^\intercal \delta) \\
&= E \left([C_3 \cdot Z^\intercal \delta] \cdot [C_1 \cdot Z^\intercal \delta - C_2 \cdot Z_E^\intercal \delta]\right) - E(C_3 \cdot Z^\intercal \delta)\cdot E(C_1 \cdot Z^\intercal \delta - C_2 \cdot Z_E^\intercal \delta) \\
&= E\left([C_3 \cdot Z^\intercal \delta] \cdot [C_1 \cdot Z^\intercal \delta - C_2 \cdot Z_E^\intercal \delta] \mid Z\right) - E\left( (C_3 \cdot Z^\intercal \delta) \mid Z\right) \cdot E\left( (C_1 \cdot Z^\intercal \delta - C_2 \cdot Z_E^\intercal \delta) \mid Z\right) \\
&= \Sigma_{11} * C_3 \cdot C_1 \cdot E(Z^\intercal Z) - \Sigma_{11} * C_3 \cdot C_2 \cdot E(Z^\intercal Z_E) - 0 \\
&= \frac{\Sigma^*_{11}}{D^\intercal (P_Z - P_{Z_E}) D}\begin{pmatrix}Z^\intercal (P_Z - P_{Z_E}) D \\ D^\intercal (P_Z - P_{Z_E}) D\end{pmatrix}
\end{align*}

Again with linear decomposition we will get to a linear reparametrization mapping $(\beta_{\rm TSLS}, \omega) \rightarrow (\beta_{\rm TSLS}, \alpha_E, \beta, u_{-E})$ with constant coefficient matrices and hence a constant Jacobian:
\begin{align*}
\omega &= \begin{pmatrix} Z^\intercal Z & Z^\intercal D \\ (Z^\intercal D)^\intercal & D^\intercal P_Z D \end{pmatrix}^{obs} \begin{pmatrix} \alpha - \alpha^* \\ \beta - \beta^* \end{pmatrix} - \begin{pmatrix} I \\ (Z^\intercal D)^\intercal (Z^\intercal Z)^{-1} \end{pmatrix}^{obs} \cdot Z^\intercal \delta + \lambda \begin{pmatrix} u \\ 0 \end{pmatrix} + \epsilon \begin{pmatrix} \alpha \\ \beta \end{pmatrix} \\
&= - \Sigma_{S,T}^{obs}\Sigma_T^{-1,obs}\cdot \beta_{\rm TSLS} + \left[\begin{pmatrix}Z^\intercal Z & Z^\intercal D \\ D^\intercal Z & D^\intercal P_{Z} D\end{pmatrix}^{obs} + \epsilon \cdot I \right] \cdot \begin{pmatrix}\alpha \\ \beta \end{pmatrix} + \lambda \begin{pmatrix} u \\ 0 \end{pmatrix} \\
&- \left[\begin{pmatrix} I \\ D^\intercal Z (Z^\intercal Z)^{-1} \end{pmatrix} \cdot Z^\intercal Y - \Sigma_{S,T}\cdot \beta_{\rm TSLS}\right]^{obs}
\end{align*}

Finally we obtain the conditional sampling density as

\begin{equation}
\begin{split}
&\ell_{\beta_0} (\beta_{\rm TSLS}, \omega \mid {\rm supp}(\alpha) = E, \text{sign}(\alpha_{E}) = \widehat{s}_E, F) \propto \phi_{\left(\beta_0, \Sigma_T^{obs}\right)} (\beta_{\rm TSLS}) \\
& \cdot g\left\{-\Sigma_{S,T}^{obs} \cdot \Sigma_T^{-1, obs} \cdot \beta_{\rm TSLS} + \left[\begin{pmatrix}Z^\intercal Z & Z^\intercal D \\ D^\intercal Z & D^\intercal P_Z D\end{pmatrix} + \epsilon\cdot I\right]^{obs}  \begin{pmatrix}\alpha \\ \beta\end{pmatrix} + \lambda \begin{pmatrix} u \\ 0 \end{pmatrix} - \left[\begin{pmatrix}Z^\intercal Y \\ D^\intercal P_Z Y\end{pmatrix}-\Sigma_{S,T}\Sigma_T^{-1, obs}\widehat{\beta}_{\rm TSLS}\right]^{obs} \right\} \\
& \cdot |\mathcal{J}|^{obs} \cdot \mathbb{I}(\mathcal{B}) 
\end{split}
\end{equation}

\end{proof}

\subsection{Asymptotic Conditional Density with AR Test Statistic}

\begin{corollary}[Asymptotic conditional distribution of sampling target] \label{thm:ar_asymp}
The asymptotic conditional density of $\widetilde{T}$ under the null hypothesis $H_0: \beta^* = \beta_0$ can be expressed (up to a proportionality constant) with respect to the variables $\widetilde{T}, \beta, \alpha_E, u_{-E}$:
\begin{equation} \label{eq:ar_cden}
\begin{split}
&\ell_{\beta_0} (\widetilde{T}, \omega \mid {\rm supp}(\alpha) = E, \text{sign}(\alpha_{E}) = \widehat{s}_E, F) \propto \phi_{\left(0, \Sigma_{T}^{obs} \right)} (\widetilde{T}) \\
& \cdot g\left\{- \Sigma_{S, T}^{obs} \Sigma_{T}^{-1,obs} \cdot \widetilde{T} + \left[\begin{pmatrix}Z^\intercal Z & Z^\intercal D \\ D^\intercal Z & D^\intercal P_Z D\end{pmatrix} + \epsilon\cdot I\right]^{obs}  \begin{pmatrix}\alpha \\ \beta\end{pmatrix} + \lambda \begin{pmatrix} u \\ 0 \end{pmatrix} - \left[\begin{pmatrix}Z^\intercal Y \\ D^\intercal P_Z Y\end{pmatrix}-\Sigma_{S, T} \Sigma_{T}^{-1}\cdot \widetilde{T}\right]^{obs} \right\} \\
& \cdot  |\mathcal{J}|^{obs} \cdot \mathbb{I}(\mathcal{B}) 
\end{split}
\end{equation}
where
\begin{align*}
\Sigma_{T} &= \hat{\Sigma}_{11} \cdot Z^\intercal (I - P_{Z_E}) Z, \quad{} \Sigma_{S,T} = \hat{\Sigma}_{11} \cdot \begin{pmatrix}Z^\intercal (I - P_{Z_E}) Z \\ D^\intercal (I - P_{Z_E}) Z\end{pmatrix}
\end{align*}
\end{corollary}

\begin{remark}[Sampling for AR statistic] \label{rmk:ar_sample}
Given samples of the above sampling target $\{\widetilde{T}_i\}_{i = 1,\ldots,N}$, we can obtain the samples of the Anderson-Rubin test statistic $\{T_{AR, i}\}_{i=1,\ldots,N}$ of the conditional density via 
\begin{equation} \label{eq:ar_target_plugin}
T_{AR} = \frac{\widetilde{T}^\intercal \cdot (Z^\intercal Z)^{-1} \cdot \widetilde{T} / (p-|E|)}{\left[(Y - D \beta_0)^\intercal \cdot (I - P_Z)\cdot (Y-D\beta_0)\right]^{obs} / (n-p)}
\end{equation}
\end{remark}

\begin{proof}[Proof of Corollary \ref{thm:ar_asymp}]

The sampling target is
\begin{align*}
\widetilde{T}_{1\times p} &= Z^\intercal (I - P_{Z_E}) (Y - D\beta_0) \\ 
&= Z^\intercal (I - P_{Z_E}) \delta \\
&= Z^\intercal \delta - (Z^\intercal Z_E) (Z^\intercal_E Z_E)^{-1} Z^\intercal_E \delta \\
&= \sqrt{n} \left\{ \frac{Z^\intercal \delta}{\sqrt{n}} - (\frac{Z^\intercal Z_E}{\sqrt{n}}) (\frac{Z_E^\intercal Z_E}{\sqrt{n}})^{-1} \frac{Z^\intercal_E \delta}{\sqrt{n}}\right\} \\
&= \sqrt{n} \left\{ \frac{Z^\intercal \delta}{\sqrt{n}} - \left[\left((\frac{Z^\intercal Z_E}{\sqrt{n}}) (\frac{Z_E^\intercal Z_E}{\sqrt{n}})^{-1}\right)^{obs} + O_p (1) \right] \frac{Z^\intercal_E \delta}{\sqrt{n}}\right\} \\
&\approx Z^\intercal \delta - \left((Z^\intercal Z_E) (Z^\intercal_E Z_E)^{-1}\right)^{obs} \cdot Z^\intercal_E \delta \\
&\stackrel{d}{=} C_1 (Z^\intercal \delta) - C_2 (Z_E^\intercal \delta)
\end{align*}
where again with the same asymptotic argument as in the proof for lemma \ref{lem:tsls_asymp} we plug in the observed values and treat the coefficient $C_2$ as constant. 

From this last line and the fact that
\begin{align*}
\frac{Z^\intercal \delta}{\sqrt{n}} &\rightarrow N(0, \Sigma^*_{11}\cdot E(Z_i^\intercal Z_i)) \\
\frac{Z^\intercal_E \delta}{\sqrt{n}} &\rightarrow N(0, \Sigma^*_{11}\cdot E(Z^\intercal_{i,E} Z_{i,E})) \\
Cov(\frac{Z^\intercal \delta}{\sqrt{n}}, \frac{Z^\intercal_E \delta}{\sqrt{n}}) &= \Sigma^*_{11} \cdot E(Z_i^\intercal Z_{i,E}) 
\end{align*}
we derive the pre-selection asymptotic distribution of $\widetilde{T}$ is gaussian with mean
$$
\mu_T = 0
$$
and approximate variance
$$
\Sigma_T = \Sigma_{11}^* \cdot Z^\intercal (I - P_{Z_E}) Z
$$

With the same KKT condition as in lemma \ref{lem:tsls_kkt_asymp}
$$
\omega = \begin{pmatrix} Z^\intercal Z & Z^\intercal D \\ (Z^\intercal D)^\intercal & D^\intercal P_Z D \end{pmatrix}^{obs} \begin{pmatrix} \alpha - \alpha^* \\ \beta - \beta^* \end{pmatrix} - \begin{pmatrix} I \\ (Z^\intercal D)^\intercal (Z^\intercal Z)^{-1} \end{pmatrix}^{obs} \cdot Z^\intercal \delta + \lambda \begin{pmatrix} u \\ 0 \end{pmatrix} + \epsilon \begin{pmatrix} \alpha \\ \beta \end{pmatrix}
$$
and the same data vector 
\begin{align*}
S &= \begin{pmatrix} I \\ (Z^\intercal D)^\intercal (Z^\intercal Z)^{-1} \end{pmatrix}^{obs} \cdot Z^\intercal \delta \\
&\stackrel{d}{=} C_3 \cdot Z^\intercal \delta
\end{align*}
and similar as in the proof of theorem \ref{thm:tsls_asymp} we see that $(\widetilde{T}, \frac{1}{\sqrt{n}}S)$ is asymptotically jointly gaussian and to derive the covariance
$$
\Sigma_{S,T} = \Sigma_{11}^* \cdot \begin{pmatrix}Z^\intercal (I - P_{Z_E}) Z \\ D^\intercal (I - P_{Z_E}) Z\end{pmatrix}
$$
With linear decomposition we obtain a linear reparametrization mapping $(\widetilde{T}, \omega) \rightarrow (\widetilde{T}, \alpha_E, \beta, u_{-E})$ with constant coefficient matrices and hence a constant Jacobian term. Eventually the conditional sampling density of the sampling target $\widetilde{T}$ is
\begin{align*}
&\ell_{\beta_0} (\widetilde{T}, \omega \mid {\rm supp}(\alpha) = E, \text{sign}(\alpha_{E}) = \widehat{s}_E, F) \propto \phi_{\left[0, \widehat{\Sigma}_{11} \cdot Z^\intercal (I - P_{Z_E}) Z \right]} (\widetilde{T}) \\
&\cdot g\left\{- \Sigma_{S, T}^{obs} \Sigma_{T}^{-1,obs} \widetilde{T} + \begin{pmatrix}Z^\intercal Z & Z^\intercal D \\ D^\intercal Z & D^\intercal P_Z D \end{pmatrix}^{obs}\begin{pmatrix}\alpha \\ \beta \end{pmatrix} - \left[\begin{pmatrix} I \\ (Z^\intercal D)^\intercal (Z^\intercal Z)^{-1} \end{pmatrix} \cdot Z^\intercal Y - \Sigma_{S, T}\Sigma_{T}^{-1} \widetilde{T}\right]^{obs} + \lambda \begin{pmatrix} u \\ 0 \end{pmatrix} + \epsilon \begin{pmatrix}\alpha \\ \beta \end{pmatrix} \right\}
\end{align*}

\end{proof}

\begin{proof}[Proof of Remark \ref{rmk:ar_sample}]
\begin{align*}
T_{AR}(\beta_0) &=\frac{(Y - D\beta_0)^\intercal (P_Z - P_{Z_E}) (Y - D\beta_0) / (p - |E|)}{(Y - D\beta_0)^\intercal P_{Z^\perp} (Y - D\beta_0)/(n - p)} \\
&\approx \frac{(Y - D\beta_0)^\intercal (P_Z - P_{Z_E}) (Y - D\beta_0) / (p - |E|)}{\left((Y - D\beta_0)^\intercal P_{Z^\perp} (Y - D\beta_0)\right)^{obs} / (n - p)} \\
&= \frac{\widetilde{T}^\intercal \cdot (Z^\intercal Z)^{-1} \cdot \widetilde{T} / (p-|E|)}{\left((Y - D\beta_0)^\intercal P_{Z^\perp} (Y - D\beta_0)\right)^{obs} / (n - p)}
\end{align*}
where in the second line we used the fact that 
\begin{align*}
(Y - D\beta_0)^\intercal P_{Z^\perp} (Y - D\beta_0)/(n - p) &= \left( E\left[(Y - D\beta_0)^\intercal P_{Z^\perp} (Y - D\beta_0)\right] + O_p (\sqrt{n}) \right)/(n - p) \\
\left((Y - D\beta_0)^\intercal P_{Z^\perp} (Y - D\beta_0)\right)^{obs} / (n - p) &= \left( E\left[(Y - D\beta_0)^\intercal P_{Z^\perp} (Y - D\beta_0)\right] + O_p (\sqrt{n}) \right)/(n - p)
\end{align*}
and hence we can essentially plug in the observed value for the denominator.
\end{proof}

\section{Sampling Strategy for Confidence Intervals}

\subsection{TSLS Statistic}

We first sample $(t_i, \alpha_i, \beta_i, u_{-E, i})_{i=1}^{N}$ at the reference value $\beta^{r} = \widehat{\beta}_{TSLS}$, and then to estimate the pivot at an arbitrary $\beta_0$ via importance sampling:
\begin{align*}
\text{pivot} (\beta_0) = \frac{\sum_{i=1}^{N} \mathbb{I} [t_i \leq T_{\beta_0}^{obs}]\cdot w(t_i, \alpha_i, \beta_i, u_{-E,i})}{\sum_{i=1}^{N} w(t_i, \alpha_i, \beta_i, u_{-E,i})}
\end{align*}
where
\begin{align*}
w(t_i,\alpha_i,\beta_i,u_{-E,i}) &= \frac{\ell_{\beta_0}(t_i, \alpha_i, \beta_i, u_{-E,i})}{\ell_{\beta^{r}}(t_i, \alpha_i, \beta_i, u_{-E,i})} \\
&= \frac{g(t_i, \alpha_i, \beta_i, u_{-E,i}, \beta_0)}{g(t_i, \alpha_i, \beta_i, u_{-E,i}, \beta^{r})}
\end{align*}
We use binary search to find the values $\beta_{U}$ and $\beta_{L}$ such that $\text{pivot} (\beta_U) =\alpha / 2$ and $\text{pivot} (\beta_{L}) = 1 - \alpha/2$ and the confidence interval will be $(\beta_L, \beta_U)$.

\subsection{TSLS Estimator}
Same importance sampling approach as above but we use $\beta_{TSLS}$ as the sampling variable.
We first sample $(t_i, \alpha_i, \beta_i, u_{-E, i})_{i=1}^{N}$ at the reference value (the observed) $\beta^{r} = \widehat{\beta}_{TSLS}$, and then to estimate the pivot at an arbitrary $\beta_0$ via importance sampling:
\begin{align*}
\text{pivot} (\beta_0) = \frac{\sum_{i=1}^{N} \mathbb{I} [t_i \leq T_{\beta_0}^{obs}]\cdot w(t_i, \alpha_i, \beta_i, u_{-E,i})}{\sum_{i=1}^{N} w(t_i, \alpha_i, \beta_i, u_{-E,i})}
\end{align*}
where
\begin{align*}
w(t_i,\alpha_i,\beta_i,u_{-E,i}) &= \frac{\ell_{\beta_0}(t_i, \alpha_i, \beta_i, u_{-E,i})}{\ell_{\beta^{r}}(t_i, \alpha_i, \beta_i, u_{-E,i})} \\
&= \frac{\phi_{\beta_0, \frac{\widehat{\Sigma}_{11}}{D^\intercal (P_Z - P_{Z_E})D}}(t_i)}{\phi_{\beta^{r},\frac{\widehat{\Sigma}_{11}}{D^\intercal (P_Z - P_{Z_E})D}}(t_i)}
\end{align*}
Again we use binary search to find the values $\beta_{U}$ and $\beta_{L}$ such that $\text{pivot} (\beta_U) =\alpha / 2$ and $\text{pivot} (\beta_{L}) = 1 -\alpha/2$ and the confidence interval will be $(\beta_L, \beta_U)$.


\section{Additional Simulation Results}

\subsection{Empirical CDF of the Conditional P-values}
We use the same data generating model as in the main manuscript and plot the empirical CDF of the conditional p-values under $H_0: \beta^* = 1$ to compare with the uniform distribution. Figure \ref{fig:tsls:null1} and \ref{fig:tsls:null2} show the result for TSLS estimator with level of endogeneity, $\Sigma^*_{12} = 0.3, 0.99$, or $\frac{\Sigma^*_{12}}{\Sigma^*_{11}\Sigma^*_{22}} = 0.3, 0.99$ since we fix $\Sigma^*_{11} = \Sigma^*_{22} = 1.$. The magnitude of $\gamma^*$ is varied by the value $r$ where $\gamma_j^* = r$ for every $j$ and $r$ ranges from $0.04$ to $2.5$.  

Figure \ref{fig:ar:null} shows the result for the AR statistic with $\Sigma^*_{12} = 0.8$, and figure \ref{fig:ar:null1} and \ref{fig:ar:null2} are with $\Sigma^*_{12} = 0.3, 0.99$.

\subsection{Confidence intervals comparison}
We use the same data generating model as before with $\rho=0.8$.

For TSLS test statistic we plot the empirical coverage rates and the average lengths of our proposed conditional interval and the naive confidence interval as a function of $r$. In Figure \ref{fig:known_0_8_gvalid_0_05}, we vary $\gamma^*$ associated with valid and invalid IVs differently, with $\gamma^*=0.05$ for valid IVs and the ratio between $\gamma^*$ for invalid over valid IVs as along the x axis. The next figure \ref{fig:known_0_8_gvalid_0_05_p100} are coverage and average lengths for the same setup but with higher dimensions $p=100, s=30$. We can see the selection effect is larger. The naive coverages are much worse. 

For the AR statistic, in Figure \ref{fig:ar:003}, we fix $\gamma^*=0.03$ for valid IVs and vary the ratio between $\gamma^*$ for invalid over valid IVs as along the x axis. in Figure \ref{fig:ar:004}, we fix $\gamma^*=0.04$ for valid IVs and vary the ratio between $\gamma^*$ for invalid over valid IVs as along the x axis. We can see again that the conditional coverage is about the guaranteed nominal level $0.95$ while naive coverage falls far below.

\begin{figure}[!ht]  
  \begin{subfigure}[b]{0.5\textwidth}
    \includegraphics[width=\textwidth]{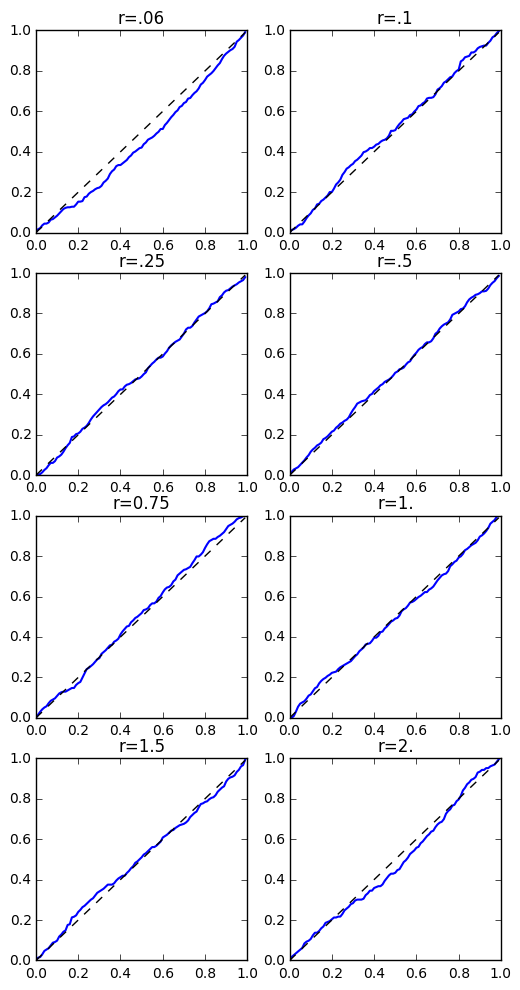}
    \caption{$\rho = 0.3$}
  \label{fig:tsls:null1}
  \end{subfigure}
  \begin{subfigure}[b]{0.5\textwidth}
    \includegraphics[width=\textwidth]{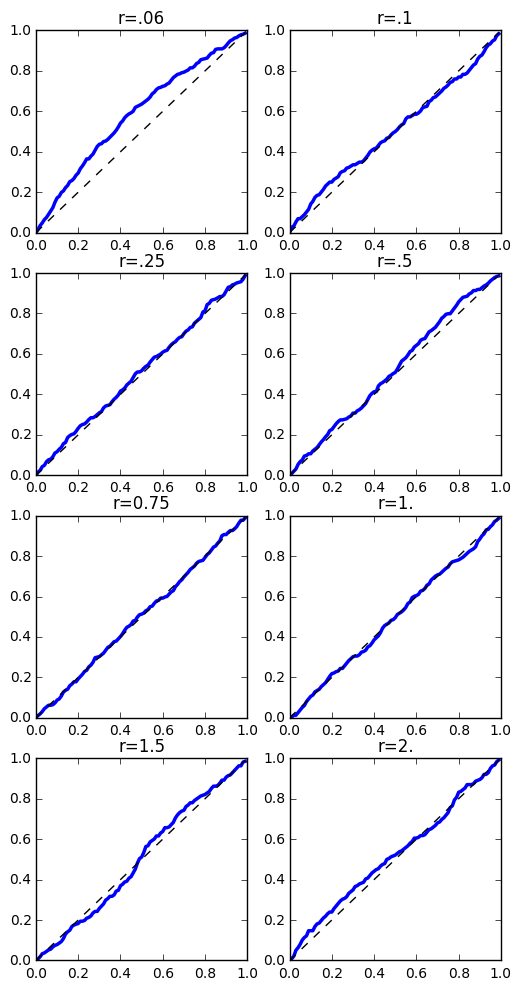}
    \caption{$\rho=0.99$}
  \label{fig:tsls:null2}
  \end{subfigure}
  \caption{TSLS: Empirical CDFs of the conditional pvalues with $\rho=0.3, 0.99$}
  \end{figure}
  
\begin{figure}[!ht]
\begin{center}
\centerline{\includegraphics[width=1.\columnwidth]{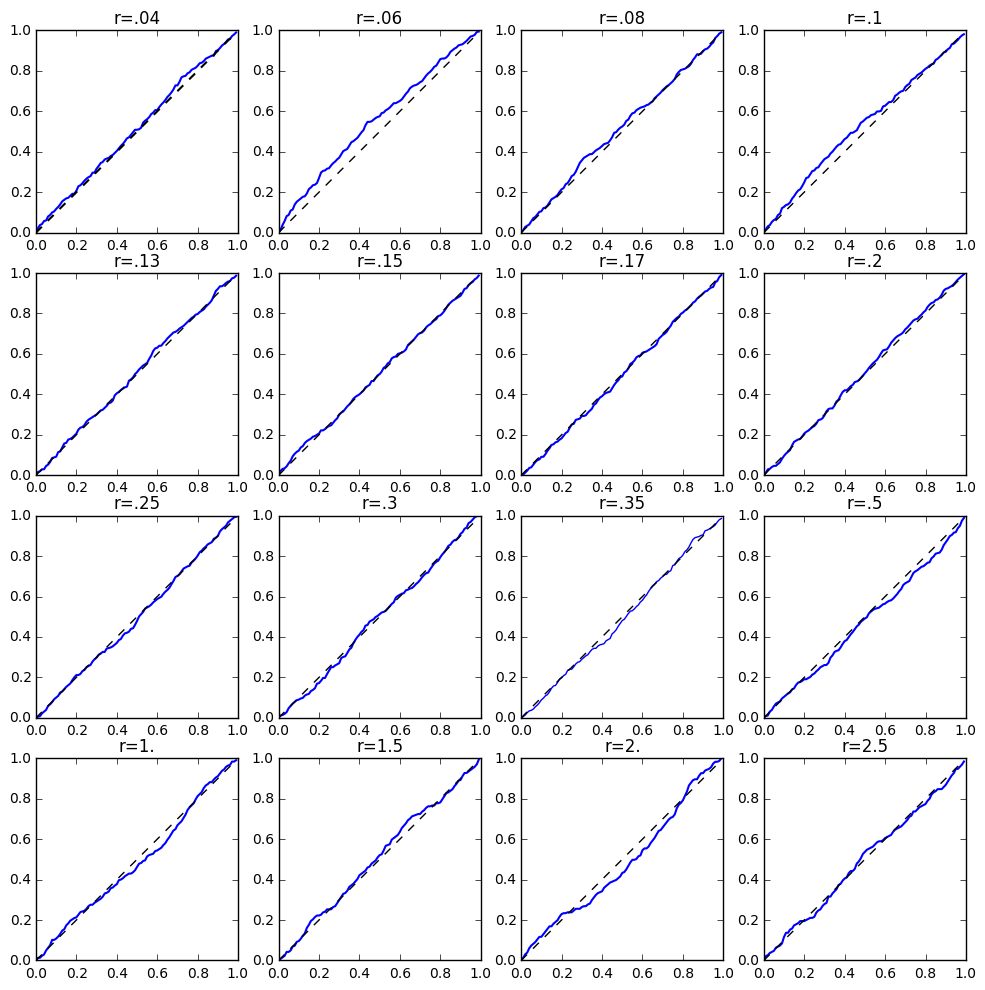}}
\caption{Comparison between the CDF of the empirical distribution of the conditional p-values using AR statistic and the uniform distribution. $r$ represents instrument strength as measured by setting $\gamma_j^* = r$ for all $j$.}
\label{fig:ar:null}
\end{center}
\end{figure} 
  
\begin{figure}[!ht]  
  \begin{subfigure}[b]{0.5\textwidth}
    \includegraphics[width=\textwidth]{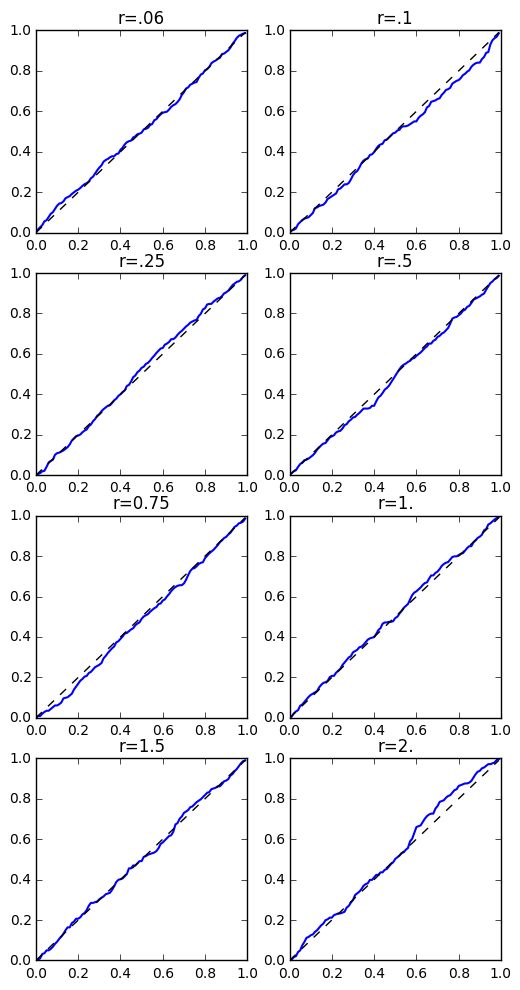}
    \caption{$\rho = 0.3$}
  \label{fig:ar:null1}
  \end{subfigure}
  \begin{subfigure}[b]{0.5\textwidth}
    \includegraphics[width=\textwidth]{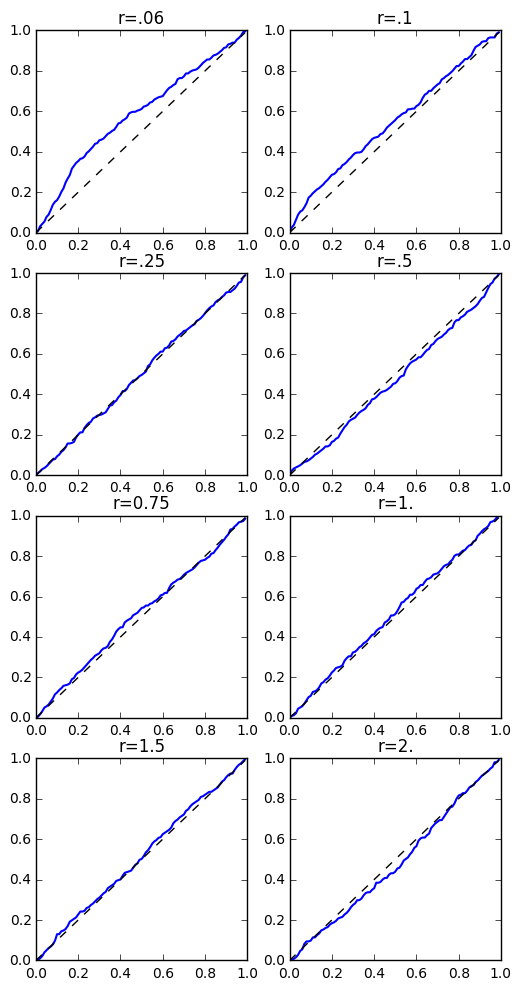}
    \caption{$\rho=0.99$}
  \label{fig:ar:null2}
  \end{subfigure}
\caption{AR: Empirical CDFs of the conditional pvalues with $\rho=0.3, 0.99$}
\end{figure}

\begin{figure}[!ht]
\begin{center}
\centerline{\includegraphics[width=1.\columnwidth]{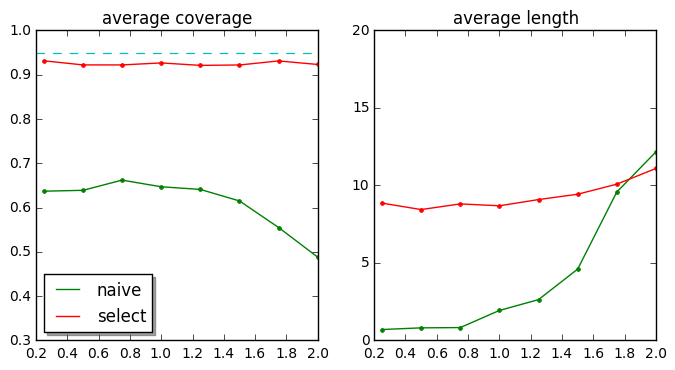}}
\caption{TSLS: Naive and conditional CI comparison. $(n,p,s)=(1000,10,3). \beta^* = 1., \alpha^*_E = 7., \rho=.8. \gamma^*=0.05$ for valid IVs and vary the ratio between $\gamma^*$ for invalid over valid IVs.}
\label{fig:known_0_8_gvalid_0_05}
\end{center}
\end{figure}

\begin{figure}[!ht]
\begin{center}
\centerline{\includegraphics[width=1.\columnwidth]{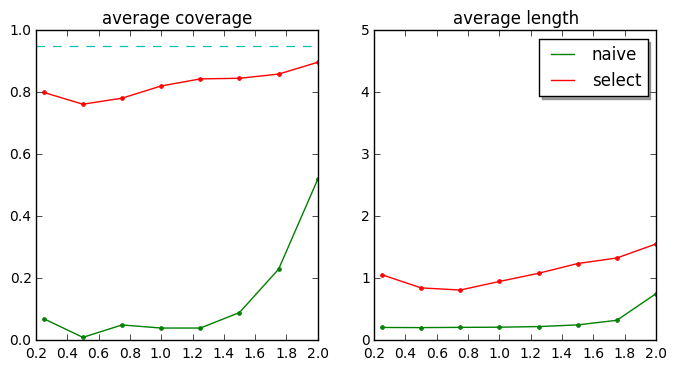}}
\caption{TSLS: Naive and conditional CI comparison. $(n,p,s)=(1000,100,30). \beta^* = 1., \alpha^*_E = 7., \rho=.8. \gamma^*=0.05$ for valid IVs and vary the ratio between $\gamma^*$ for invalid over valid IVs.}
\label{fig:known_0_8_gvalid_0_05_p100}
\end{center}
\end{figure}

\begin{figure}[!ht]
  \begin{subfigure}[b]{0.5\textwidth}
    \includegraphics[width=\textwidth]{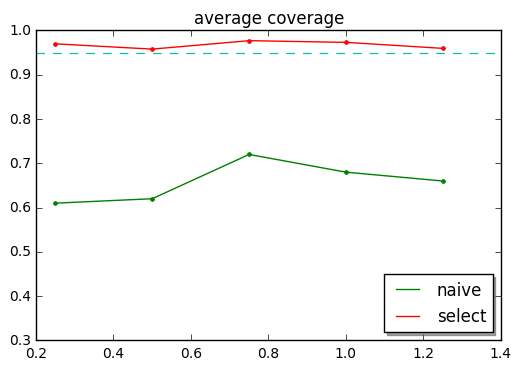}
    \caption{$\gamma^* = 0.03$ for valid IVs.}
    \label{fig:ar:003}
  \end{subfigure}
  \begin{subfigure}[b]{0.5\textwidth}
    \includegraphics[width=\textwidth]{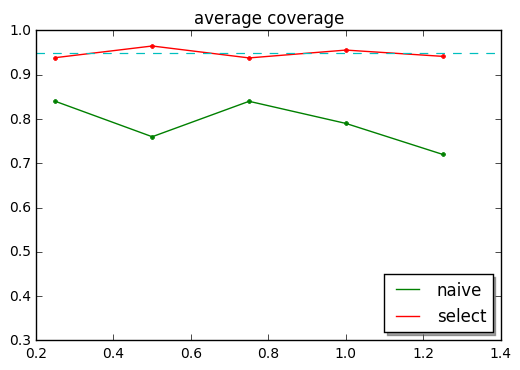}
    \caption{$\gamma^*=0.04$ for valid IVs.}
    \label{fig:ar:004}
  \end{subfigure}
\caption{AR: Naive and conditional CI coverage comparison. $(n,p,s)=(1000,10,3). \beta^*=1., \alpha^*_E = 7., \rho=.8$. Vary the ratio between $\gamma^*$ for invalid over valid IVs.}
\end{figure}

\clearpage
\section{Inference with Summary Statistics}

\begin{lemma} \label{lem:sumstat} Suppose we have the summary statistics $\widehat{\beta}_{j,Y}, {\rm se}^2 (\widehat{\beta}_{j,Y}), \widehat{\beta}_{j,D}, {\rm se}^2 (\widehat{\beta}_{j,D})$ and the data $Y,D,Z$ follows the defined model. If 
\begin{itemize}
\item[(i)] $Z^\intercal Z$ is a diagonal, but not necessarily identity, matrix and 
\item[(ii)] $Y, D, Z$ are centered to mean zero,
\end{itemize} 
the quantities $Z^\intercal Y, Z^\intercal D, Z^\intercal Z$ can be written as
\begin{align*}
Z_{(j)}^\intercal Z_{(j)} = \frac{(n-1)\cdot se^2 (\widehat{\beta}_{1,Y}) + \widehat{\beta}_{1,Y}^2}{(n-1)\cdot {\rm se}^2 (\widehat{\beta}_{i,Y}) + \widehat{\beta}_{i,Y}^2} \cdot c \\
Z_{(j)}^\intercal Y = \widehat{\beta}_{i,Y} \cdot \frac{(n-1)\cdot {\rm se}^2 (\widehat{\beta}_{1,Y}) + \widehat{\beta}_{1,Y}^2}{(n-1)\cdot {\rm se}^2 (\widehat{\beta}_{i,Y}) + \widehat{\beta}_{i,Y}^2} \cdot c \\
Z_{(j)}^\intercal D = \widehat{\beta}_{i,D} \cdot \frac{(n-1)\cdot {\rm se}^2 (\widehat{\beta}_{1,Y}) + \widehat{\beta}_{1,Y}^2}{(n-1)\cdot {\rm se}^2 (\widehat{\beta}_{i,Y}) + \widehat{\beta}_{i,Y}^2} \cdot c
\end{align*}
where $c$ is an unknown constant factor, i.e. they are known up to an unknown factor $c$.
\end{lemma}

\begin{theorem}[Inference with summarized data] \label{thm:sum_infer} Under the assumptions in Lemma \ref{lem:sumstat},  the selective inference density in \eqref{eq:tsls_cden} is identical, regardless of the value of $c$.
\end{theorem}

\begin{proof}[Proof of Lemma \ref{lem:sumstat}]
First consider the summary statistics between $Y$ and $Z_{(j)}$. The formula for the OLS coefficient and standard error are

\begin{align*}
\hat{\beta}_{i,Y} &= \frac{Z_{(i)}^\intercal Y}{Z^\intercal_{(i)} Z_{(i)}} \\
se(\hat{\beta}_{i,Y}) &= \sqrt{s^2 \cdot (Z_{(i)}^\intercal Z_{(i)})^{-1}} = \sqrt{\frac{1}{n - 1} \frac{Y^\intercal(I-H_{Z_{(i)}})Y}{Z_{(i)}^\intercal Z_{(i)}}} \\
&= \sqrt{\frac{1}{n - 1} \left[\frac{Y^\intercal Y}{Z_{(i)}^\intercal Z_{(i)}} - (\frac{Z_{(i)}^\intercal Y}{Z_{(i)}^\intercal Z_{(i)}})^2 \right]} = \sqrt{\frac{1}{n - 1} (\frac{Y^\intercal Y}{Z_{(i)}^\intercal Z_{(i)}} - \hat{\beta}_{i,Y}^2)}
\end{align*}
where $H_{Z_{(i)}} = \frac{Z_{(i)} Z_{(i)}^\intercal}{Z_{(i)}^\intercal Z_{(i)}}$.
Combine the two to get $Z_{(i)} Y$ and $Y^\intercal Y$ in terms of $Z_{(i)}^\intercal Z_{(i)}$ with $n$ sample size:
\begin{align*}
Z_{(i)}^\intercal Y &= \hat{\beta}_{i,Y} \cdot (Z_{(i)}^\intercal Z_{(i)}) \\
Y^\intercal Y &= \left((n - 1)\cdot se^2 (\hat{\beta}_{i,Y}) + \hat{\beta}_{i,Y}^2\right)\cdot (Z_{(i)}^\intercal Z_{(i)}) = k_i \cdot (Z_{(i)}^\intercal Z_{(i)}) \label{eq:YY}
\end{align*}
Denote $c = Z_{(1)}^\intercal Z_{(1)}$, then $Y^\intercal Y = k_1 \cdot c$, and that $Z_{(i)}^\intercal Z_{(i)} = \frac{k_1}{k_i} \cdot c$, $Z_{(i)}^\intercal Y = \hat{\beta}_i \frac{k_1}{k_i} \cdot c$.

Similarly for $D$ and $Z_i$, we have 
$$
Z_{(i)}^\intercal D = \hat{\beta}_{i, D} \cdot (Z_{(i)}^\intercal Z_{(i)}) = \hat{\beta}_{i, D} \cdot \frac{k_1}{k_i} \cdot c
$$
as the formulae given in the lemma.

\end{proof}

\begin{proof}[Proof of Theorem \ref{thm:sum_infer}]
We are going to show the conditional density in Corollary \ref{thm:tsls_beta_asymp} is the same no matter the value of $c$. Let us consider the several components in the density.

\begin{itemize}
\item Select $E$ from SisVive procedure: \\
\begin{align*}
&\min_{\hat{\alpha}, \hat{\beta}} \frac{1}{2} ||P_Z (Y - Z\alpha - D\beta)||^2_2 + \lambda ||\alpha||_1 \\
=& \frac{1}{2} (Y - Z\alpha - D \beta)^\intercal Z (Z^\intercal Z)^{-1} Z^\intercal (Y - Z \alpha - D\beta) + \lambda ||\alpha||_1 \\
=& \frac{1}{2} (Z^\intercal Y - (Z^\intercal Z)\alpha - Z^\intercal D \beta)^\intercal (Z^\intercal Z)^{-1} (Z^\intercal Y - (Z^\intercal Z)\alpha - Z^\intercal D \beta) + \lambda ||\alpha||_1
\end{align*}

$\lambda \sim c$ based on the way it was chosen \citep{negahban_unified_2012}, and the $\ell_2$ term also has the factor $c$, hence the entire optimization problem is the same up to $c$, and the selection result would be the same, i.e. $E$ will be the same.

Also the selection event $\mathcal{B}$ will also be the same based on the same $E$.

Specifically, the optimization problem can be solved as a Lasso with response $(Z^\intercal Z)^{-\frac{1}{2}} Z^\intercal Y$, the design $\begin{pmatrix}(Z^\intercal Z)^{\frac{1}{2}} & (Z^\intercal Z)^{-\frac{1}{2}} Z^\intercal D\end{pmatrix}$ and coefficients $\begin{pmatrix}\alpha \\ \beta \end{pmatrix}$.

\item Estimate $\widehat{\beta}_{TSLS}$ \\
\begin{align*}
\widehat{\beta}_{TSLS} &= \frac{D^\intercal (P_Z - P_{Z_E}) Y}{D^\intercal (P_Z - P_{Z_E}) D} \\
&= \frac{(Z^\intercal D)^\intercal (Z^\intercal Z)^{-1} (Z^\intercal Y) - (Z_E^\intercal D)^\intercal (Z_E^\intercal Z_E)^{-1} (Z_E^\intercal Y)}{(Z^\intercal D)^\intercal (Z^\intercal Z)^{-1} (Z^\intercal D) - (Z_E^\intercal D)^\intercal (Z_E^\intercal Z_E)^{-1} (Z_E^\intercal D)} \\
&= \frac{\sum_i \frac{Z_i^\intercal D \cdot Z_i^\intercal Y}{Z_i^\intercal Z_i} - \sum_{i\in E} \frac{Z_i^\intercal D \cdot Z_i^\intercal Y}{Z_i^\intercal Z_i}}{\sum_i \frac{(Z_i^\intercal D)^2}{Z_i^\intercal Z_i} - \sum_{i \in E} \frac{(Z_i^\intercal D)^2}{Z_i^\intercal Z_i}}
\end{align*}
where $Z^\intercal Z$ is diagonal and plugging in formulae in lemma \ref{lem:sumstat} we can see that $\widehat{\beta}_{TSLS}$ is unchanged with $c$.

\item Estimate of covariance term $\hat{\Sigma}_{11}$ \\
Consider the reduced-form model
\begin{align*}
Y &= Z_E (\beta^* \gamma_E^* + \alpha_E^*) + Z_{-E} (\gamma^*_{-E} \beta^*)+ \xi_1 \\
D &= Z \gamma^* + \xi_2
\end{align*}
where $(\xi_{i1}, \xi_{i2}) \sim N(0, \Omega^*)$. 
As common in MR, we also assume that $\Omega_{12}^* = 0$, 
and the reduced form model is no different than a pair of OLS regression models in terms of estimating the covariance $\Omega^*$:
$$
\hat{\Omega} = \begin{pmatrix} s^2_{YZ} & 0 \\ 0 & s^2_{DZ}\end{pmatrix}
$$
where 
\begin{align*}
s^2_{YZ} &= \frac{1}{n-p+1} \sum_i (Y_i - \hat{\beta}_{i,Y} Z_i)^2 = \frac{1}{n-p+1} [Y^\intercal Y - \sum_i \frac{(Z_i^\intercal Y)^2}{Z_i^\intercal Z_i}] \\
s^2_{DZ} &= \frac{1}{n-p+1} \sum_i (D_i - \hat{\beta}_{i,D} Z_i)^2 = \frac{1}{n-p+1} [D^\intercal D - \sum_i \frac{(Z_i^\intercal D)^2}{Z_i^\intercal Z_i}] \\
\end{align*}
where from standard OLS regression formulae we know that with orthogonal design matrix $Z$ the individual coefficients of the marginal OLSs are the same as the coefficients in the multivariate OLS with all the $Z_i$'s together.

From the one-to-one mapping and plug in a consistent estimator of $\beta^*$ such as $\widehat{\beta}_{TSLS}$, we can obtain a consistent estimate of $\Sigma^*$
$$
\widehat{\Sigma} = \begin{pmatrix} 1 & -\widehat{\beta}_{TSLS} \\ 0 & 1 \end{pmatrix} \widehat{\Omega} \begin{pmatrix} 1 & 0 \\ -\widehat{\beta}_{TSLS} & 1 \end{pmatrix}
$$

$\widehat{\Sigma}$ has the factor $c$.

\item Covariance terms: \\
The covariance of the asymptotic distribution of TSLS estimator
$$
Cov(\widehat{\beta}_{TSLS}) = \frac{\widehat{\Sigma}_{11}}{D^\intercal (P_Z - P_{Z_E}) D} = \frac{\widehat{\Sigma}_{11}}{D^\intercal Z (Z^\intercal Z)^{-1} Z^\intercal D - D^\intercal Z_E (Z_E^\intercal Z_E)^{-1} Z_E^\intercal D}
$$
the factor $c$ cancels with the numerator and denominator and hence $Cov(\widehat{\beta}_{TSLS})$ does not change with $c$.

\begin{align*}
\Sigma_{S,T} &= \frac{\widehat{\Sigma}_{11}}{D^\intercal (P_Z - P_{Z_E}) D} \begin{pmatrix}Z^\intercal (P_Z -P_{Z_E})D \\ D^\intercal (P_Z - P_{Z_E}) D \end{pmatrix} \\
&= Cov(\widehat{\beta}_{TSLS}) \cdot \begin{pmatrix}Z^\intercal D - Z^\intercal Z_E (Z_E^\intercal Z_E)^{-1} Z_E^\intercal D \\ D^\intercal Z (Z^\intercal Z)^{-1} Z^\intercal D - D^\intercal Z_E (Z_E^\intercal Z_E)^{-1} Z_E^\intercal D \end{pmatrix}
\end{align*}
$\Sigma_{S,T}$ has the factor $c$.

\item Selective density \\
\begin{align*}
&\ell_{\beta_0} (\beta_{TSLS}, \alpha, \beta, u_{-E}) \approx \phi_{\beta_0, \frac{\widehat{\Sigma}_{11}}{D^\intercal (P_Z - P_{Z_E})D}} (\beta_{TSLS}) \cdot \mathbb{I}(\mathcal{B}) \\
&\cdot g\left\{-\Sigma_{S,T}^{obs} \cdot \beta_{TSLS} + \left(\begin{pmatrix}Z^\intercal Z & Z^\intercal D \\ (Z^\intercal D)^\intercal & D^\intercal P_Z D\end{pmatrix}+ \epsilon I\right)^{obs}\begin{pmatrix}\alpha \\ \beta \end{pmatrix} + \lambda \begin{pmatrix} u \\ 0 \end{pmatrix} - \left(\begin{pmatrix}Z^\intercal Y \\ D^\intercal P_Z Y\end{pmatrix} - \Sigma_{S,T} \cdot \beta_{TSLS}\right)^{obs}\right\} \\
\end{align*}
The way we set the randomization $g$ scales with $c$, as well as $\lambda$, and the terms inside $g$ also has the factor $c$, hence overall $g(\cdot)$ and also $\phi (\cdot)$ both have the same value as $c$ changes.

\end{itemize}

\end{proof}

\section{Application: a Developmental Economics Dataset} \label{sec:application}

It is of interest to study the relationship between income and expenditures on goods and services in developmental economics, and specifically the efficient wage hypothesis, which suggests that the increased income would lead to workers being better fed, and eventually result in better worker productivity, especially in developing economies. We are going to apply our method to a real dataset in development economics, as a follow-up and comparison to the analysis in \citet{kang2015simple}, which had been analyzed in \citet{bouis1990agricultural} and \citet{bouis1992estimates}.

The dataset has $n=405$ observations of Philippine farm households, and the outcome $Y_i$ is household food expenditures, the treatment $D_i$ is household's log income. We are interested in analyzing the effect of income $D_i$ on demand for food $Y_i$. There are $L=4$ instrument candidates, namely cultivated area per capita $Z_{i1}$, worth of assets $Z_{i2}$, binary indicator of household electricity $Z_{i3}$, and the house flooring quality $Z_{i4}$. There are also covariates $X_j$'s, namely mother's education, father's education, mother's age, father's age, mother's nutritional knowledge, price of corn, price of rice, population density of the municipality, and number of household members in adult equivalents (c.f. page 82 of \citet{bouis1990agricultural}). 

In terms of the general model, there will be exogenous variable $X_j$'s in the structural model
\begin{align*}
\begin{array}{r@{\mskip\thickmuskip}l}
Y_i &=  Z_{i}^T \alpha^* +  X_i^T \kappa + D_i \beta^* + \delta_i \\
D_i &=  Z_{i}^T \gamma^* + X_i^T \pi + \xi_{i2} 
\end{array}
\end{align*}
we can replace the variables $Y_i, D_i, Z_i$ with the residuals after regressing them on $X$, i.e. replace $Y, D, Z$ by $\tilde{Y} = (I - P_X) Y, \tilde{D} = (I - P_X) D, \tilde{Z} = (I - P_X) Z$ (Wang and Zivot 1998), hence to get rid of the dependence on nuisance parameters $(\kappa, \pi)$ and the resulting sufficient statistic does not involve $X$ for $\beta^*$. The rest of the analysis proceeds exactly the same with $\tilde{Y}, \tilde{D}, \tilde{Z}$.

In the original setup of \citet{kang2015simple}, there is a parameter $U$ as the sensitivity parameter and be provided by the researcher as an upper bound on the number of invalid instruments in the model, i.e. $U > s$, with smaller $U$ representing a larger confidence that most instruments are valid and vice versa. In our analysis, we will also report confidence intervals for different given values of $U$, to better compare with the other analysis.

In below table \ref{tab:selective}, $U=1$ means it is believed to be no invalid IVs and the confidence interval is the same as the classical one; for $U=2$ our procedures always picks $Z_2$ as invalid and the resulting intervals with TSLS or AR statistic are shown; for $U=3$ our procedures always picks $Z_2$, and the first lines are intervals when picking $Z_1$ additionally, and the second lines are intervals otherwise. 
In table \ref{tab:kang}, the intervals from \cite{kang2015simple} are shown with corresponding $U$ values. On the high level, their intervals is constructed to be the union of all the intervals given possible enumerations of the invalid IVs with $s<U$.

Generally speaking, when $U > 1$ meaning there is possibly invalid instruments in the model, the selective confidence intervals tend to be longer than those merely based on TSLS statistics, and the difference tends to be larger when $U$ increases, i.e. when there might be many invalid instruments. Because our approach conditions on the actual selection result, there will be more than one possible confidence interval based on the observed selection result. This could be expected because of the randomized nature of our method, just like the results are also based on the random model chosen by data splitting. Our method is valuable when the researcher trusts in the data-driven method to detect the invalid IVs, and the model will only be based on the selection result. It can also be seen that even with the union of our intervals for a given $U$, our interval is slightly shorter than theirs, possible because of the power increase from our randomized procedure.

\begin{table}[!ht]
\caption {Selective confidence intervals} \label{tab:selective} 
\begin{center}
\begin{tabular} {|c|c|c|c|}
 \hline
 Test & $U=1$ (Naive) & $U=2$ & $U=3$ \\
 \hline
 TSLS & (0.043, 0.053) & $(0.034, 0.058)$ & $(-0.007, 0.012)$ \\
 &&&$(0.037, 0.064)$\\
 \hline
 AR & (0.044, 0.054) & $(0.042, 0.051)$  & $(-0.015, 0.031)$ \\
 &&&$(0.038, 0.063)$\\
\hline
\end{tabular}
\end{center}
\end{table}

\begin{table}[!ht]
\caption {\cite{kang2015simple} confidence intervals} \label{tab:kang} 
\begin{center}
\begin{tabular}{|c|c|c|c|}
\hline
 Test & $U=1$ (Naive) & $U=2$ & $U=3$ \\
 \hline
 TSLS & $(0.043, 0.053)$ & $(0.031, 0.059)$ & $(-0.017, 0.064)$ \\
 \hline
  AR & $(0.044, 0.054)$ & $(0.037, 0.058)$  & $(-0.027, 0.068)$  \\
\hline
\end{tabular}
\end{center}
\end{table}

\begin{remark}
If we directly fit an OLS of $D$ on $Z$, the t-values are $8.163, 8.232, 1.080, 2.774$ and we can see that the most significant one is the 2nd IV which is always picked up first. When it then picks up the 1st IV, the remaining IVs are relatively weakly correlated with $D$, and the resulting confidence intervals contain 0. When it then picks up either the 3rd or the 4th IV, the remaining IVs (including the 1st IV) are more correlated with $D$, and the resulting confidence intervals do not contain 0 and much resemble those for $U=2$ case.
But of course these uncertainly come from the data itself, any conditional inference has to deal with this ambiguity of the validity of instruments. Selective inference is trying to eliminate the selection bias on top of this.
\end{remark}

\bibliographystyle{apa}
\bibliography{selectIV}

\end{document}